\documentclass{lmcs} 
\pdfoutput=1

\usepackage{lastpage}
\lmcsdoi{18}{2}{20}
\lmcsheading{}{\pageref{LastPage}}{}{}%
{Oct.~28,~2021}{Jun.~15,~2022}{}

\usepackage[utf8]{inputenc}

\allowdisplaybreaks

\keywords{Shapley value, inconsistent databases, functional dependencies, database repairs}

\usepackage{hyperref}
\theoremstyle{plain} 

\usepackage{amssymb}
\usepackage{cite}
\usepackage{calc}
\usepackage{tikz}
\usepackage{multicol}
\usetikzlibrary{positioning,chains,fit,shapes,calc}
\usepackage{subcaption}

\usepackage{algorithm,algorithmicx}
\usepackage[noend]{algpseudocode}

\usepackage{ifmtarg}
\makeatletter
\newcommand{\toprule}{\hrule height.8pt depth0pt \kern2pt} 
\newcommand{\midrule}{\kern2pt\hrule\kern2pt} 
\newcommand{\bottomrule}{\kern2pt\hrule\relax}
\newcommand{\algcaption}[2][]{%
  \refstepcounter{algorithm}%
  \@ifmtarg{#1}
    {\addcontentsline{loa}{figure}{\protect\numberline{\thealgorithm}{\ignorespaces #2}}}
    {\addcontentsline{loa}{figure}{\protect\numberline{\thealgorithm}{\ignorespaces #1}}}%
  \toprule
  \textbf{\fname@algorithm~\thealgorithm}\ #2\par 
  \midrule
}
\makeatother

\def\ester#1{{\color{red} #1}}

\def\e#1{\emph{#1}}
\def\scs{\mathbf{S}}

\def\pall{\mathord{\mathcal{P}}}

\def\ra{\rightarrow}

\def\shapley{\mathrm{Shapley}}

\newcommand*{\MyDef}{\mathrm{def}}
\newcommand*{\eqdefU}{\ensuremath{\mathop{\overset{\MyDef}{=}}}}
\newcommand*{\eqdef}{\mathrel{\overset{\MyDef}{\resizebox{\widthof{\eqdefU}}{\heightof{=}}{=}}}}

\def\set#1{\mathord{\{#1\}}}

\def\dl{\mathrel{{:}{\text{-}}}}

\def\prev{\mathsf{prev}}

\def\prob{\mathsf{val}}
\def\root{r}

\def\probf{\mathsf{val}'}

\def\assn{\mathrel{{:}{=}}}

\newcommand{\pluseq}{\mathrel{+}=}

\def\Pr#1{\mathord{\mathrm{Pr}\left[#1\right]}}

\def\val#1{\mathtt{#1}}

\def\att#1{\mathrm{#1}}

\newcommand{\algname}[1]{{\sf #1}}
\def\myrulewidth{3.20in}

\def\therule{\rule{\myrulewidth}{0.2pt}}

\newenvironment{insidecode}[3]
{
\begin{tabular}{p{\myrulewidth}}
\multicolumn{1}{c}{\rule{0mm}{3mm}{\bf #3} $\algname{#1}(\mbox{#2})$\vspace{-0.6em}}\\
\therule\vskip-0.8em\therule
\vspace{0em}
\begin{algorithmic}[1]}
{\end{algorithmic}
\vskip-0.3em\therule
\end{tabular}}

\newenvironment{subroutine}
{\begin{algorithm}\floatname{algorithm}{Subroutine}}
{\end{algorithm}}

\newcounter{subroutine}

\newenvironment{msubroutine}[2]
{
\stepcounter{subroutine}
\addtocounter{algorithm}{-1}
\renewcommand{\thealgorithm}{\arabic{subroutine}} 
\floatname{algorithm}{Subroutine}
\algcaption{#1}\label{#2}
\begin{algorithmic}[1]
}
{
\end{algorithmic}
\bottomrule
}

\newenvironment{malgorithm}[2]
{
\algcaption{#1}\label{#2}
\begin{algorithmic}[1]
}
{
\end{algorithmic}
\bottomrule
}

\newcommand{\eat}[1]{}

\newenvironment{repeatresult}[2]
{\vskip0.5em\par\textbf{#1 #2.}\em}
{\vskip1em}

\def\fpsharpp{\mathrm{FP}^{\mathrm{\#P}}}

\def\I{\mathcal{I}}

\def\MI{\mathsf{MI}}
\def\MC{\mathsf{MC}}
\def\MR{\mathsf{MR}}

\def\Id{\I_{\mathsf{d}}}
\def\Imr{\I_{\mathsf{R}}}

\def\Imi{\I_{\MI}}
\def\Imc{\I_{\MC}}
\def\Ip{\I_{\mathsf{P}}}

\newcommand{\depset}{\mathrm{\Delta}}

\def\match{\mathsf{M}}

\newenvironment{citeddefinition}[1]
{\begin{defiC}[{\cite{#1}}]}
{\end{defiC}}

\def\Exp{\mathbb{E}}

\algnewcommand{\NULL}{\textsc{null}}

\algnewcommand{\IfThenElse}[3]{
  \State \algorithmicif\ #1\ \algorithmicthen\ #2\ \algorithmicelse\ #3}
  
  \algnewcommand{\IfThen}[2]{
  \State \algorithmicif\ #1\ \algorithmicthen\ #2}

\def\eqed{\hfill$\Diamond$}

\begin{document}

\title[Shapley Value of Inconsistency Measures]{The Shapley Value of Inconsistency Measures\texorpdfstring{\\}{}
  for Functional Dependencies}

\author[E.~Livshits]{Ester Livshits}	
\address{Technion, Haifa, Israel}	
\email{\{esterliv,bennyk\}@cs.technion.ac.il}  

\author[B.~Kimelfeld]{Benny Kimelfeld}	





\begin{abstract}
  \noindent    Quantifying the inconsistency of a database is motivated by various
  goals including reliability estimation for new datasets and progress
  indication in data cleaning. Another goal is to attribute to
  individual tuples a level of responsibility to the overall
  inconsistency, and thereby prioritize tuples in the explanation or
  inspection of errors.  Therefore, inconsistency quantification and
  attribution have been a subject of much research in knowledge
  representation and, more recently, in databases. As in many other
  fields, a conventional responsibility sharing mechanism is the
  Shapley value from cooperative game theory.  In this article, we
  carry out a systematic investigation of the
  complexity of the Shapley value in common inconsistency measures for
  functional-dependency (FD) violations. For several measures we
  establish a full classification of the FD sets into tractable and
  intractable classes with respect to Shapley-value computation. We
  also study the complexity of approximation in intractable cases.
\end{abstract}

\maketitle


\section{Introduction}

Inconsistency measures for knowledge bases have received considerable attention from the Knowledge Representation (KR) and Logic communities~\cite{DBLP:conf/ijcai/KoniecznyLM03,DBLP:journals/jolli/Knight03,DBLP:conf/kr/HunterK06,DBLP:journals/jiis/GrantH06,DBLP:conf/kr/HunterK08,DBLP:journals/ai/HunterK10,DBLP:journals/ijar/GrantH17,DBLP:journals/ki/Thimm17}. More recently, inconsistency measures have also been studied from the database viewpoint~\cite{DBLP:conf/sum/Bertossi18,DBLP:conf/sigmod/LivshitsKTIKR21}.  Such measures quantify the extent to which the database violates a set of integrity constraints. There are multiple reasons why one might be using such measures.  For one, the measure can be used for estimating the usefulness or reliability of new datasets for data-centric applications such as business intelligence~\cite{DBLP:conf/atal/CholvyPRT15}. Inconsistency measures have also been proposed as the basis of progress indicators for data-cleaning systems~\cite{DBLP:conf/sigmod/LivshitsKTIKR21}.  Finally, the measure can be used for attributing to individual tuples a level of responsibility to the overall inconsistency~\cite{10.1093/jigpal/exr002 ,DBLP:conf/uai/Thimm09},  thereby prioritize tuples in the explanation/inspection/correction of errors.

{
  \begin{figure}[t]
  \centering
\def\arraystretch{1.1}
\begin{tabular}{c||c|c|c|c|c}\hline
fact  & $\att{train}$ & $\att{departs}$ & $\att{arrives}$ & $\att{time}$ & $\att{duration}$\\\hline\hline
$f_1$ & $\val{16}$ & $\val{NYP}$ & $\val{BBY}$ & $\val{1030}$ & $\val{315}$\\
$f_2$ & $\val{16}$ & $\val{NYP}$ & $\val{PVD}$ & $\val{1030}$ & $\val{250}$\\
$f_3$ & $\val{16}$ & $\val{PHL}$ & $\val{WIL}$ & $\val{1030}$ & $\val{20}$\\
$f_4$ & $\val{16}$ & $\val{PHL}$ & $\val{BAL}$ & $\val{1030}$ & $\val{70}$\\
$f_5$ & $\val{16}$ & $\val{PHL}$ & $\val{WAS}$ & $\val{1030}$ & $\val{120}$\\
$f_6$ & $\val{16}$ & $\val{BBY}$ & $\val{PHL}$ & $\val{1030}$ & $\val{260}$\\
$f_7$ & $\val{16}$ & $\val{BBY}$ & $\val{NYP}$ & $\val{1030}$ & $\val{260}$\\
$f_8$ & $\val{16}$ & $\val{BBY}$ & $\val{WAS}$ & $\val{1030}$ & $\val{420}$\\
$f_9$ & $\val{16}$ & $\val{WAS}$ & $\val{PVD}$ & $\val{1030}$ & $\val{390}$\\\hline
\end{tabular}
\caption{\label{fig:trains} The inconsistent database of our running example.}
\end{figure}
}

\begin{exa}\label{example:introduction}
 Figure~\ref{fig:trains} depicts an inconsistent
  database
  that stores a train schedule.
  For example, the tuple $f_1$
  states that train number 16 will depart from the New York Penn Station at time 1030 and arrive at the Boston Back Bay Station after 315 minutes. Assume that we have the functional dependency stating that
  the train number and departure time determine the departure station. All tuples in the database are involved in violations of this constraint, as they all agree on the train number and departure time, but there is some disagreement on the departure station. Hence, one can argue that every fact in the database affects the overall level of inconsistency in the database. But how should we measure the \e{responsibility} of the tuples to this inconsistency? For example, which of the tuples $f_1$ and $f_3$ has a greater contribution to inconsistency? To this end, we can adopt some conventional concepts for responsibility sharing, and in this article we study the computational aspects involved in the measurement of those. 
\eqed
\end{exa}

A conventional approach to dividing the responsibility for a quantitative property (here an inconsistency measure) among entities (here the database tuples) is the \e{Shapley value}~\cite{shapley:book1952}, which is a game-theoretic formula for wealth distribution in a cooperative game. The Shapley value has been applied in a plethora of domains, including economics~\cite{gul1989bargaining}, law~\cite{nenova2003value}, environmental science~\cite{petrosjan2003time,liao2015case}, social network analysis~\cite{narayanam2011shapley}, physical network analysis~\cite{ma2010internet}, and advertisement~\cite{DBLP:conf/www/BesbesDGIS19}.  In data management, the Shapley value has been used for determining the relative contribution of features in machine-learning predictions~\cite{DBLP:conf/ijcai/LabreucheF18,DBLP:conf/nips/LundbergL17}, the responsibility of tuples to database queries~\cite{DBLP:conf/pods/ReshefKL20,DBLP:conf/icdt/LivshitsBKS20,DBLP:journals/jdiq/BertossiG20},
and the reliability of data sources~\cite{DBLP:conf/atal/CholvyPRT15}.

The Shapley value has also been studied in a context similar to the one we adopt in this article---assigning a level of inconsistency to statements in an inconsistent knowledge base~\cite{DBLP:journals/ai/HunterK10,DBLP:conf/ijcai/YunVCB18, 10.1093/jigpal/exr002, DBLP:conf/uai/Thimm09}.
Hunter and Konieczny~\cite{DBLP:conf/kr/HunterK06,DBLP:journals/ai/HunterK10,DBLP:conf/kr/HunterK08} use the maximal Shapley value of one inconsistency measure in order to define a new inconsistency measure.
Grant and Hunter~\cite{DBLP:conf/ecsqaru/GrantH15} considered information systems distributed along data sources of different reliabilities,
and apply the Shapley value to determine the expected blame of each statement to the overall inconsistency.
Yet, with all the investigation that has been conducted on the Shapley value of inconsistency,
we are not aware of any results or efforts regarding the computational complexity of calculating this value.

\begin{exa}
Let us define the following cooperative game over the database of Figure~\ref{fig:trains}. We have nine players---the tuples of the database. One of the measures that we consider for quantifying the level of inconsistency of a coalition of players is the number of tuple pairs in this group that violate the constraints. For example, consider the constraint defined in Example~\ref{example:introduction}. The inconsistency level of the group $\{f_1,f_3,f_5\}$ is $2$, as there are two conflicting tuple pairs: $\{f_1,f_3\}$ and $\{f_1,f_5\}$. The inconsistency level of the entire database is $29$, as this is the total number of conflicting pairs in the database. The Shapley value allows us to measure the contribution of each individual tuple to the overall inconsistency level. For example, the Shapley value of the tuple $f_1$, in this case, will be lower than the Shapley value of the tuple $f_3$ (we will later show how this value is computed), which indicates that $f_3$ has a higher impact on the inconsistency than $f_1$.
\end{exa}

In this work, we embark on a systematic analysis of the complexity of the Shapley value of database tuples relative to inconsistency measures, where the goal is to calculate the contribution of a tuple to inconsistency. Our main results are summarized in Table~\ref{table:complexity}.  We consider inconsistent databases with respect to functional dependencies (FDs), and basic measures of inconsistency following Bertossi~\cite{DBLP:conf/lpnmr/Bertossi19} and Livshits, Ilyas, Kimelfeld and Roy~\cite{DBLP:conf/sigmod/LivshitsKTIKR21}.
We note that these measures are all adopted from the measures studied in the aforementioned KR research.
In our setting, an individual tuple affects the inconsistency of only its containing relation, since the constraints are FDs. Hence, our analysis focuses on databases with a single relation; in the end of each relevant section, we discuss the generalization to multiple relations. While most of our results easily extend to multiple relations, some extensions require a more subtle proof.

  \begin{table}[t]
  \renewcommand{\arraystretch}{1.2}
  \centering
  \caption{The complexity of the (exact ; approximate) Shapley value of different inconsistency measures.
    \label{table:complexity}}
\begin{tabular}{c||c|c|c}
        & \textbf{lhs chain} & \textbf{no lhs chain, PTime c-repair} & \textbf{other}\\\hline\hline
        $\Id$ & PTime & \multicolumn{2}{c}{$\fpsharpp$-complete ;  FPRAS} \\ \hline
        $\Imi$ & \multicolumn{3}{c}{ PTime } \\ \hline
        $\Ip$ & \multicolumn{3}{c}{ PTime } \\ \hline
        $\Imr$ & PTime & \textbf{?} ; FPRAS &
                             NP-hard~\cite{DBLP:journals/tods/LivshitsKR20} ; no FPRAS \\ \hline
        $\Imc$ & PTime & \multicolumn{2}{c}{$\fpsharpp$-complete~\cite{DBLP:conf/pods/LivshitsK17} ; \textbf{?}} \\ \hline
    \end{tabular}
  \end{table}

More formally, we investigate the following computational problem for any fixed combination of
a relational signature, 
a set of FDs, and an inconsistency measure: given a database and a tuple, compute the Shapley value of the tuple with respect to the inconsistency measure. As Table~\ref{table:complexity} shows, two of these measures are computable in polynomial time: $\Imi$ (number of FD violations) and $\Ip$ (number of problematic facts that participate in violations).  For two other measures, we establish a full dichotomy in the complexity of the Shapley value: $\Id$ (the drastic measure---0 for consistency and 1 for inconsistency) and $\Imc$ (number of maximal consistent subsets, a.k.a.~repairs). The dichotomy in both cases is the same: when the FD set has, up to equivalence, an lhs chain (i.e., the left-hand sides form a chain w.r.t.~inclusion~\cite{DBLP:conf/pods/LivshitsK17}), the Shapley value can be computed in polynomial time; in any other case, it is $\fpsharpp$-hard (hence, requires at least exponential time under conventional complexity assumptions). In the case of $\Imr$ (the minimal number of tuples to delete for consistency), the problem is solvable in polynomial time in the case of an lhs chain, and NP-hard whenever it is intractable to find a cardinality repair~\cite{DBLP:journals/tods/LivshitsKR20}; however, the problem is open for every FD set in between, for example, the bipartite matching constraint $\set{A\ra B,B\ra A}$.

We also study the complexity of approximating the Shapley value
and show the following (as described in Table~\ref{table:complexity}). First, in the case of $\Id$, there is a (multiplicative) fully polynomial-time approximation scheme (FPRAS) for every set of FDs. In the case of $\Imc$, approximating the Shapley value of \e{any} intractable (non-lhs-chain) FD set is at least as hard as approximating the number of maximal matchings of a bipartite graph---a long standing open problem~\cite{DBLP:journals/corr/abs-1807-04803}. In the case of $\Imr$, we establish a full dichotomy, namely FPRAS vs.~hardness of approximation, that has the same separation as the problem of finding a cardinality repair.

This article is the full version of a conference publication~\cite{DBLP:conf/icdt/LivshitsK21}.
We have added all of the
proofs, intermediate results and algorithms that were excluded from the 
conference version. In particular, we have included in this version the proofs of Observation~\ref{obs:reductionexp}, Lemma~\ref{lem:drastic-hard}, Lemma~\ref{lemma:imr_help}, Lemma~\ref{lemma:mc1}, and Lemmea~\ref{lemma:mc2}, and the algorithms of Figures~\ref{alg:DrasticShapleyF},~\ref{alg:MRShapleyF}, and~\ref{alg:MCShapleyF}.
Furthermore, the results of the conference publication have been restricted to schemas with a single relation symbol. While some of the results (e.g., all of the lower bounds) immediately generalize to schemas with multiple relation symbols, some generalizations (in particular, the upper bounds for $\Id$ and $\Imc$) require a more subtle analysis that we provide in this article. We generalize the upper bounds for all the measures to schemas with multiple relation symbols, in the corresponding sections.

The rest of the article is organized as follows. After presenting the basic notation and terminology in Section~\ref{sec:preliminaries}, we formally define the studied problem and give initial observations
in Section~\ref{sec:shap}. In Section~\ref{sec:tractable}, we describe polynomial-time algorithms for $\Imi$ and $\Ip$.
Then, we explore the measures $\Id$, $\Imr$ and $\Imc$ in Sections~\ref{sec:drastic}, \ref{sec:minrep} and~\ref{sec:counting}, respectively. We conclude and discuss future directions in Section~\ref{sec:conclusions}.

\section{Preliminaries}\label{sec:preliminaries}

We begin with preliminary concepts and notation that we use throughout the article. 

\subsection {Database Concepts.}
By a \e{relational schema} we refer to a sequence $(A_1,\dots,A_n)$ of attributes. A database $D$ over
$(A_1,\dots,A_n)$ is a finite set of tuples, or \e{facts}, of the form $(c_1,\dots,c_n)$, where each $c_i$ is a constant
from a countably infinite domain.  For a fact $f$ and an attribute $A_i$, we denote by $f[A_i]$ the value associated by $f$ with the attribute $A_i$ (that is, $f[A_i]=c_i$). Similarly, for a sequence $X=(A_{j_1},\dots,A_{j_m})$ of attributes, we denote by $f[X]$ the tuple $(f[A_{j_1}],\dots,f[A_{j_m}])$. Generally, we use letters from the beginning of the English alphabet (i.e., $A,B,C,...$) to denote single attributes and letters from the end of the alphabet (i.e., $X,Y,Z,...$) to denote sets of attributes.
We may omit stating the relational schema
of a database $D$ when it is clear from the context or
irrelevant.

A \e{Functional Dependency} (FD for short) over $(A_1,\dots,A_n)$ is an expression of the form $X\rightarrow Y$, where $X,Y\subseteq\set{A_1,\dots,A_m}$.  We may also write the attribute sets $X$ and $Y$ by concatenating the attributes
(e.g., $AB\rightarrow C$ instead of $\set{A,B}\rightarrow\set{C}$).  A database $D$ satisfies $X\rightarrow Y$ if every two facts $f,g\in D$ that agree on the values of the attributes of $X$ also agree on the values of the attributes of $Y$ (that is, if $f[X]=g[X]$ then $f[Y]=g[Y]$). A database $D$ \e{satisfies} a set $\depset$ of FDs, denoted by $D\models\depset$, if $D$ satisfies every FD of $\depset$. Otherwise, $D$ \e{violates} $\depset$ (denoted by $D\not\models\depset$). Two FD sets over the same relational schema are \e{equivalent} if every database that satisfies one of them also satisfies the other.

Let $\depset$ be a set of FDs and $D$ a database (which may violate $\depset$).  A \e{repair} (\e{of $D$ w.r.t.~$\depset$}) is a maximal consistent subset of $D$; that is, $E\subseteq D$ is a repair if $E\models\depset$ but $E'\not\models\depset$ for every $E\subsetneq E'$. A \e{cardinality repair} (or \e{c-repair} for short) is a repair of maximum cardinality; that is, it is a repair $E$ such that $|E|\ge |E'|$ for every repair $E'$.

\begin{exa}\label{example:train-intro}
 Consider again the database of Figure~\ref{fig:trains}
  over the relational schema
  $$(\att{train},\att{departs},\att{arrives},\att{time},\att{duration}).$$
The FD set $\depset$ consists of the two FDs:
$$\circ\,\,\,\, \att{train}\,\,\att{time}\rightarrow \att{departs}
\quad\quad\quad
\circ\,\,\,\,\att{train}\,\,\att{time}\,\,\att{duration}\rightarrow \att{arrives}$$
The first FD states that the departure station is determined by the train number and departure time, and the
second FD states that the arrival station is determined by the train number, the departure time, and the duration of the ride.

Observe that the database of Figure~\ref{fig:trains} violates the FDs as all the facts refer to the same train number and departure time, but there is
no agreement on the departure station. Moreover, the facts $f_6$ and $f_7$ also agree on the duration, but disagree on the arrival station. The database has five repairs: \e{(a)} $\set{f_1,f_2}$, \e{(b)} $\set{f_3,f_4,f_5}$,
\e{(c)} $\set{f_6,f_8}$, \e{(d)} $\set{f_7,f_8}$, and \e{(e)} $\set{f_9}$; only the second one is a cardinality repair.
\eqed
\end{exa}

\subsection{Shapley Value.}
A \e{cooperative game} of a set $A$ of players is a function $v:\pall(A)\rightarrow \mathbb{R}$, where $\pall(A)$ is the power set of $A$, such that $v(\emptyset)=0$. The value $v(B)$ should be thought of as the joint wealth obtained by the players of $B$ when they cooperate.  The \e{Shapley value} of a player $a\in A$ measures the contribution of $a$ to the total wealth $v(A)$ of the game~\cite{shapley:book1952}, and is formally defined by
$$\shapley(A,v,a)\eqdef \dfrac{1}{|A|!}\sum_{\sigma \in \Pi_A} (v(\sigma_a\cup \{a\})-v(\sigma_a))$$
where $\Pi_A$ is the set of all permutations over the players of $A$ and $\sigma_a$ is the set of players that appear before $a$ in the permutation $\sigma$.
Intuitively, the Shapley value of a player $a$ is the expected contribution of $a$ to a subset constructed by drawing players randomly one by one (without replacement), where the contribution of $a$ is the change to the value of $v$ caused by the addition of $a$.
An alternative formula for the Shapley value, that we will use in this article, is the following.
$$\shapley(A,v,a)\eqdef \sum_{B\subseteq
  A\setminus\set{a}}\frac{|B|!\cdot (|A|-|B|-1)!}{|A|!}
\Big(v(B\cup\set{a})-v(B)\Big)$$
Observe that $|B|!\cdot (|A|-|B|-1)!$ is the number of permutations where the players of $B$ appear first, then $a$, and then the rest of the players.

\subsection{Complexity.}
In this article, we focus on the standard notion of \e{data complexity}, where the relational schema and set of FDs are considered fixed and the input consists of a database and a fact. In particular, a polynomial-time algorithm may be exponential in the number of attributes or FDs. Hence, each combination of a relational schema and an FD set defines a distinct problem, and different combinations may have different computational complexities.
We discuss both exact and approximate algorithms for computing Shapley values. 

Recall that a \e{Fully-Polynomial Randomized Approximation Scheme} (FPRAS, for short) for a function $f$ is a randomized algorithm $A(x,\epsilon,\delta)$ that returns an $\epsilon$-approximation of $f(x)$ with probability at least $1-\delta$, given an input $x$ for $f$ and $\epsilon,\delta\in(0,1)$, in time polynomial in $x$, $1/\epsilon$, and $\log(1/\delta)$. Formally, an FPRAS, satisfies:
\[\Pr{{f(x)}/{(1+\epsilon)}\leq A(x,\epsilon,\delta)\leq
    (1+\epsilon)f(x)}\geq 1-\delta\,.\]
Note that this notion of FPRAS refers to a \e{multiplicative} approximation, and we adopt this notion implicitly unless stated otherwise. We may also
write ``multiplicative'' explicitly for stress.
In cases where the function $f$ has a bounded range, it also makes sense to discuss 
an \e{additive} FPRAS where $\Pr{f(x)-\epsilon\leq A(x,\epsilon,\delta)\leq f(x)+\epsilon}\geq
  1-\delta$.
We refer to an additive FPRAS, and explicitly state so, in cases where the Shapley value is in the range $[0,1]$.

\section{The Shapley Value of Inconsistency Measures}\label{sec:shap}
In this article, we study the Shapley value of facts with respect to measures of database inconsistency. More precisely, the cooperative game that we consider here is determined by an inconsistency measure $\I$, and the facts of the database take the role of the players.
In turn, an \e{inconsistency measure} $\I$ is a function that maps pairs $(D,\depset)$ of a database $D$ and a set $\depset$ of FDs to a number $\I(D,\depset)\in [0,\infty)$. Intuitively, the higher the value $\I(D,\depset)$ is, the more inconsistent (or, the less consistent) the database $D$ is w.r.t.~$\depset$.
The Shapley value of a fact $f$ of a database $D$ w.r.t.~an FD set $\depset$ and inconsistency measure $\I$ is then defined as follows.
\begin{equation}\label{eq:shapley}
   \shapley(D,\depset,f,\I)\eqdef \sum_{E\subseteq
  (D\setminus\set{f})}\hskip-0.5em\frac{|E|!\cdot (|D|-|E|-1)!}{|D|!}
\Big(\I(E\cup\set{f},\depset)-\I(E,\depset)\Big)
\end{equation}
We note that the definition of the Shapley value requires the cooperative game to be zero on the empty set~\cite{shapley:book1952}
and this is indeed the case for all of the inconsistency measures $\I$ that we consider in this work.
Next, we introduce each of these measures.

\begin{itemize}
    \item $\Id$  is the \e{drastic measure} that takes the value $1$ if the database is inconsistent and the value $0$ otherwise~\cite{DBLP:journals/ki/Thimm17}.
    \item $\Imi$ counts the \e{minimal inconsistent subsets}~\cite{DBLP:conf/kr/HunterK08,DBLP:journals/ai/HunterK10}; in the case of FDs, these subsets are simply the pairs of tuples that jointly violate an FD.
    \item $\Ip$ is the number of \e{problematic facts}, where a fact is problematic
      if it belongs to a minimal inconsistent subset~\cite{DBLP:conf/ecsqaru/GrantH11}; in the case of FDs, a fact is problematic if and only if it
      participates in a pair of facts that jointly violate $\depset$.      
    \item $\Imr$ is the minimal number of facts that we need to delete from the database for $\depset$ to be satisfied (similarly to the concept of a cardinality repair and proximity in Property Testing)~\cite{DBLP:conf/ecsqaru/GrantH13,DBLP:journals/jacm/GoldreichGR98,DBLP:conf/lpnmr/Bertossi19}.
    \item $\Imc$ is the number of \e{maximal consistent subsets} (i.e., repairs)~\cite{DBLP:conf/ecsqaru/GrantH11,DBLP:journals/ijar/GrantH17}.
\end{itemize}

Table~\ref{table:complexity} summarizes the complexity results for the different measures. The first column (lhs chain) refers to FD sets that have a left-hand-side chain---a notion that was introduced by Livshits et al.~\cite{DBLP:conf/pods/LivshitsK17}, and we recall in the next section. The second column
(no lhs chain, PTime c-repair) refers to FD sets that do not have a left-hand-side chain, but entail a polynomial-time cardinality repair computation according to the dichotomy of Livshits et al.~\cite{DBLP:journals/tods/LivshitsKR20} that we discuss in more details in Section~\ref{sec:minrep}.

\begin{exa}
  Consider again the database of our running example. Since the database is inconsistent w.r.t.~the FD set defined in Example~\ref{example:train-intro}, we have that $\Id(D,\depset)=1$. As for the measure $\Imi$, the reader can easily verify that there are twenty nine pairs of tuples that jointly violate the FDs; hence, we have that $\Imi(D,\depset)=29$. Since each tuple participates in at least one violation of the FDs, it holds that $\Ip(D,\depset)=9$. Finally, as we have already seen in Example~\ref{example:train-intro}, the database has five repairs and a single cardinality repair obtained by deleting six facts. Thus, $\Imr(D,\depset)=6$ and $\Imc(D,\depset)=5$. In the next sections, we discuss the computation of the Shapley value for each one of these measures.
  \eqed\end{exa}

\paragraph*{Preliminary analysis.}
We study the \e{data complexity} of computing $\shapley(D,\depset,f,\I)$ for different inconsistency measures $\I$. To this end, we give here two important observations that we will use throughout the article. The first observation is that the computation of $\shapley(D,\depset,f,\I)$ can be easily reduced to the computation of the expected value of the inconsistency measure over all 
Here, we denote by $\Exp_{D'\sim U_m(D\setminus\set{f})}\big(\I(D'\cup\set{f},\depset)\big)$ the expected value of $\I(D'\cup\set{f},\depset)$ over all subsets $D'$ of $D\setminus\set{f}$ of a given size $m$, assuming a uniform distribution. Similarly, $\Exp_{D'\sim U_m(D\setminus\set{f})}\big(\I(D',\depset)\big)$ is the expected value of $\I(D',\depset)$ over all such subsets $D'$.

\def\obsreduction{
Let $\I$ be an inconsistency measure. The following holds.
{\small
\[\shapley(D,\depset,f,\I)=
    \frac{1}{|D|}\sum_{m=0}^{|D|-1} \left[ \Exp_{D'\sim U_m(D\setminus\set{f})}\big(\I(D'\cup\set{f},\depset)\big)- \Exp_{D'\sim U_m(D\setminus\set{f})}\big(\I(D',\depset)\big)\right]\]
}}

\begin{obs}\label{obs:reductionexp}
\obsreduction
\end{obs}
\begin{proof}
  We have the following.
  \begin{align}
  &\shapley(D,\depset,f,\I)=\sum_{D'\subseteq
      (D\setminus\set{f})}\frac{|D'|!  (|D|-|D'|-1)!}{|D|!}
    \Big(\I(D'\cup\set{f},\depset)-\I(D',\depset)\Big)\notag\\
    &=\sum_{m=0}^{|D|-1}\underset{\substack{D'\subseteq
      (D\setminus\set{f}) \\ |D'|=m}}{\sum}\frac{m!  (|D|-m-1)!}{|D|!}
    \Big(\I(D'\cup\set{f},\depset)-\I(D',\depset)\Big)\notag\\
    &=\sum_{m=0}^{|D|-1} \frac{m!  (|D|-m-1)!}{|D|!}{{|D|-1}\choose m} \underset{\substack{D'\subseteq
      (D\setminus\set{f}) \\ |D'|=m}}{\sum}
    \frac{1}{{{|D|-1}\choose m}}\Big(\I(D'\cup\set{f},\depset)\Big)\label{eq:uniform}\\
    &- \sum_{m=0}^{|D|-1}\frac{m!  (|D|-m-1)!}{|D|!}{{|D|-1}\choose m} \underset{\substack{D'\subseteq
      (D\setminus\set{f}) \\ |D'|=m}}{\sum}
    \frac{1}{{{|D|-1}\choose m}}\Big(\I(D',\depset)\Big)\notag\label{eq:exp}\\
    &=\sum_{m=0}^{|D|-1} \frac{m!  (|D|-m-1)!}{|D|!}{{|D|-1}\choose m} \Exp_{D'\sim U_m(D\setminus\set{f})}\big(\I(D'\cup\set{f},\depset)\big)\\
    &- \sum_{m=0}^{|D|-1}\frac{m!  (|D|-m-1)!}{|D|!}{{|D|-1}\choose m} \Exp_{D'\sim U_m(D\setminus\set{f})}\big(\I(D',\depset)\big)\notag\\
        &=\frac{1}{|D|}\sum_{m=0}^{|D|-1} \left[\Exp_{D'\sim U_m(D\setminus\set{f})}\big(\I(D'\cup\set{f},\depset)\big)
    -  \Exp_{D'\sim U_m(D\setminus\set{f})}\big(\I(D',\depset)\big)\right]\notag
\end{align}
Note that in Equation~\eqref{eq:uniform} we multiply and divide by the value ${{|D|-1}\choose m}$. The expectation expression of Equation~\eqref{eq:exp} is due to the fact that $1/{{{|D|-1}\choose m}}$ is the probability of a random subset of size $m$ of $D\setminus\set{f}$ in the uniform distribution.
\end{proof}

Observation~\ref{obs:reductionexp} implies that to compute the Shapley value of $f$, it suffices to compute the expectations of the amount of inconsistency over subsets $D'$ and $D'\cup\set{f}$, where $D'$ is drawn uniformly from the space of subsets of size $m$, for every $m$.
  More precisely, the computation of the Shapley value is Cook reducible\footnote{Recall that a \e{Cook reduction} from a function $F$ to a function $G$ is a
    polynomial-time \e{Turing reduction} from $F$ to $G$, that is, an algorithm that computes $F$ with an oracle to a solver of $G$.}
    to the computation of these expectations. Our algorithms will, indeed, compute these expectations instead of the Shapley value.

The second observation is the following. One of the basic properties of the Shapley value is one termed ``efficiency''---the sum of the Shapley values over all the players equals the total wealth~\cite{shapley:book1952}. This property implies that
$\sum_{f\in D}\shapley(D,\depset,f,\I)=\I(D,\depset)$.
Thus, whenever the measure itself is computationally hard, so is the Shapley value of facts. 

\begin{fact}\label{fact:measuretoshap}
  Let $\I$ be an inconsistency measure. The computation of $\I$  is Cook reducible to the computation of the Shapley value of facts under $\I$.
\end{fact}

This observation can be used for showing lower bounds on the complexity of the Shapley value, as we will see in the next sections.

\section{Measures $\Imi$ and $\Ip$: The Tractable Measures}\label{sec:tractable}

We start by discussing two tractable measures, namely $\Imi$ and $\Ip$. We first give algorithms for computing the Shapley value for these measures, and then discuss the generalization to multiple relations.

\subsection{Computation}
Recall that $\Imi$ counts the pairs of facts that jointly violate at least one FD.
An easy observation is that a fact $f$ increases the value of the measure $\Imi$ by $i$ in a permutation $\sigma$ if and only if $\sigma_f$ contains exactly $i$ facts that are in conflict with $f$. Hence, assuming that $D$ contains $N_f$ facts that conflict with $f$, the Shapley value for this measure can be computed in the following way:
\begin{align*}
&\shapley(D,\depset,f,\Imi)=\sum_{E\subseteq
  (D\setminus\set{f})}\hskip-0.5em\frac{|E|!\cdot (|D|-|E|-1)!}{|D|!}
\Big(\I(E\cup\set{f},\depset)-\I(E,\depset)\Big)\\
&=\frac{1}{|D|!}\sum_{i=1}^{N_f}\underset{\substack{E\subseteq
  (D\setminus\set{f}) \\ |E\cap N_f|=i}}\sum\hskip-0.5em
|E|!\cdot (|D|-|E|-1)!\cdot i
=\frac{1}{|D|!}\sum_{i=1}^{N_f}\sum_{m=i}^{|D|-1}\underset{\substack{E\subseteq
  (D\setminus\set{f}) \\ |E|=m \\ |E\cap N_f|=i}}\sum\hskip-0.5em
m!\cdot (|D|-m-1)!\cdot i\\
&=\frac{1}{|D|!}\sum_{i=1}^{N_f}\sum_{m=i}^{|D|-1}{N_f \choose i}{|D|-N_f-1 \choose m-i}\cdot m!\cdot (|D|-m-1)!\cdot i
\end{align*}
Therefore, we immediately obtain the following result.

\begin{thm}
  Let $\depset$ be a set of FDs.  $\shapley(D,\depset,f,\Imi)$ is computable in polynomial time, given $D$ and $f$.
\end{thm} 

We now move on to $\Ip$ that counts the ``problematic'' facts; that is, facts that participate in a violation of $\depset$.
Here, a fact $f$ increases the measure by $i$ in a permutation $\sigma$ if and only if $\sigma_f$ contains precisely $i-1$ facts that are in conflict with $f$, but not in conflict with any other fact of $\sigma_f$ (hence, all these facts and $f$ itself are added to the group of problematic facts).
 We prove the following.
 
 \begin{thm}\label{thm:ponefd}
Let $\depset$ be a set of FDs. $\shapley(D,f,\depset,\Ip)$ is computable in polynomial time, given $D$ and $f$.
\end{thm}
\begin{proof}
We now show how the expected values of Observation~\ref{obs:reductionexp} can be computed in polynomial time.
We start with $\Exp_{D'\sim U_m(D\setminus\set{f})}\big(\Ip(D',\depset)\big)$.
We consider the uniform distribution $U_m(D\setminus\set{f})$
over the subsets of size $m$ of $D\setminus\set{f}$.
We denote by $X$ the random variable holding the number of
problematic facts in the random subset.
We denote by $Y_g$ the random variable that holds $1$ if the fact
$g$ is in the random subset and, moreover, participates there in a violation of the FDs.
In  addition, we denote the expectations of these variables by $\Exp(X)$ and
$\Exp(Y_g)$, respectively (without explicitly stating the distribution
$D'\sim U_m(D\setminus\set{f})$ in the
subscript). 
Due to the linearity of the expectation we have:

\begin{align*}
  \Exp_{D'\sim U_m(D\setminus\set{f})}\big(\Ip(D',\depset)\big)&=\Exp(X)=\Exp\left(\sum_{g\in {D\setminus\set{f}}}Y_g\right)=\sum_{g\in {D\setminus\set{f}}}\Exp(Y_g)
\end{align*}
Hence, the computation of $\Exp_{D'\sim U_m(D\setminus\set{f})}\big(\Ip(D',\depset)\big)$ reduces to the computation of $\Exp(Y_g)$, and this value can be computed as follows.

\begin{align*}
  \Exp(Y_g)&=\Pr{g\mbox{ is selected}}\times
  \Pr{\mbox{a conflicting fact is selected}\mid g\mbox{ is selected}}\\
  &=\frac{{{|D|-2}\choose {m-1}}}{{{|D|-1}\choose m}}\cdot \frac{\sum_{k=1}^{N_g}{{N_g}\choose k}\cdot {{|D|-1-N_g}\choose{m-k-1}}}{{{|D|-2}\choose {m-1}}}=\frac{\sum_{k=1}^{N_g}{{N_g}\choose k}\cdot {{|D|-1-N_g}\choose{m-k-1}}}{{{|D|-1}\choose m}}
\end{align*}
where $N_g$ is the number of facts in $D\setminus\set{f}$ that are in conflict with $g$.

We can similarly consider the distribution $U_m(D\setminus\set{f})$ and show that the expectation 
$\Exp_{D'\sim U_m(D\setminus\set{f})}\big(\Ip(D'\cup\set{f},\depset)\big)$ is equal to $\sum_{g\in {D\setminus\set{f}}}\Exp(Y'_g)$,
where $Y'_g$ is a random variable that holds $1$ if $g$ is selected in the random subset and, moreover, participates in a violation of the FDs, and $0$ otherwise.
For a fact $g$ that is not in conflict with $f$ it holds that
$\Exp(Y'_g)=\Exp(Y_g)$,
while for a fact $g$ that is in conflict with $f$ it holds that
\begin{align*}
  \Exp(Y'_g)&=\Pr{g\mbox{ is selected}}={{|D|-2}\choose {m-1}}/{{{|D|-1}\choose m}}\,.
    \tag*{\qedhere}
\end{align*}
\end{proof}

\subsection{Generalization to Multiple Relations}
 The results of this section immediately generalize to schemas with multiple relation symbols. This is true since one of the basic properties of the Shapley value is linearity~\cite{shapley:book1952}:
$$\shapley(D,f,\depset,a\cdot\alpha+b\cdot\beta)=a\cdot\shapley(D,f,\depset,\alpha)+b\cdot\shapley(D,f,\depset,\beta)$$
and both measures, $\Imi$ and $\Ip$, are additive over multiple relations, that is, the value of the measure on the entire database is the sum of the values over the individual relations.

\section{Measure $\Id$: The Drastic Measure}\label{sec:drastic}

In this section, we consider the drastic measure $\Id$.  While the measure itself is extremely simple and, in particular, computable in polynomial time (testing whether $\depset$ is satisfied), it might be intractable to compute the Shapley value of a fact. In particular, we prove a dichotomy for this measure, classifying FD sets into ones where the Shapley value can be computed in polynomial time and the rest where the problem is $\fpsharpp$-complete.\footnote{Recall that $\fpsharpp$ is the class of polynomial-time functions with an oracle to a problem in \#P (e.g., count the satisfying assignments of a propositional formula).}

\subsection{Dichotomy}
Before giving our dichotomy, we recall the definition of a \e{left-hand-side chain} (lhs chain, for short), introduced by Livshits et al.~\cite{DBLP:conf/pods/LivshitsK17}.

\begin{citeddefinition}{DBLP:conf/pods/LivshitsK17}
An FD set $\depset$ has a left-hand-side chain if for every two FDs $X\rightarrow Y$ and $X'\rightarrow Y'$ in $\depset$, either $X\subseteq X'$ or $X'\subseteq X$. 
\end{citeddefinition}

\begin{exa}
  The FD set of our running example (Example~\ref{example:train-intro}) has an lhs chain.
We could also define $\depset$ with redundancy by adding the following FD: $\att{train}\,\,\att{time}\,\,\att{arrives}\rightarrow \att{departs}$.
The resulting FD set does not have an lhs chain, but it is \e{equivalent} to an FD set with an lhs chain. An example of an FD set that does not have an lhs chain,
not even up to equivalence, is $\set{\att{train}\,\,\att{time}\rightarrow\att{departs}, \att{train}\,\,\att{departs}\rightarrow\att{time}}$.
\eqed\end{exa}

We prove the following.

\def\theoremdrastic{
Let $\depset$ be a set of FDs.
  If $\depset$ is equivalent to an FD set with an lhs chain, then $\shapley(D,f,\depset,\Id)$ is computable in polynomial time, given $D$ and $f$. Otherwise, the problem is $\fpsharpp$-complete.}

\begin{thm}\label{thm:drastic}
  \theoremdrastic
\end{thm}

Interestingly, this is the exact same dichotomy that we obtained in prior work~\cite{DBLP:conf/pods/LivshitsK17} for the problem of counting subset repairs. We also showed that this tractability criterion is decidable in polynomial time by computing a minimal cover: if $\depset$ is equivalent to an FD set with an lhs chain, then every minimal cover of $\depset$ has an lhs chain.
In the remainder of this section, we prove
Theorem~\ref{thm:drastic}.

\subsubsection{Hardness Side.}
The proof of the hardness side of Theorem~\ref{thm:drastic} has two steps. We first show hardness
for the matching constraint $\set{A\rightarrow B,B\rightarrow A}$ over the schema $(A,B)$, and this proof is 
similar to the proof of 
Livshits et al.~\cite{DBLP:conf/icdt/LivshitsBKS20} for the problem of computing the Shapley contribution of facts to the result of the query $q()\dl R(x),S(x,y),T(y)$.
Then, from this case to the remaining cases we apply the \e{fact-wise reductions} that have been devised in prior work~\cite{DBLP:conf/pods/LivshitsK17}.
We start by proving hardness for  $\set{A\rightarrow B, B\rightarrow A}$.

\begin{figure}
  \centering
  \input{bipartite.pspdftex}
  \caption{The databases constructed in the reduction of the proof of
  Lemma~\ref{lem:drastic-hard}.}
  \label{fig:bipartite}
  \end{figure}

\def\drastichard{
Computing $\shapley(D,f,\depset,\Id)$ for the FD set $\depset=\set{A\rightarrow B, B\rightarrow A}$ over the relational schema $(A,B)$ is $\fpsharpp$-complete.
}

\begin{lem}\label{lem:drastic-hard}
\drastichard
\end{lem}
\begin{proof}
  We construct a reduction from the problem of computing the number $|\match(g)|$ of matchings in a bipartite graph $g$~\cite{DBLP:journals/siamcomp/Valiant79}. Note that we consider partial matchings; that is, subsets of edges that consist of mutually-exclusive edges.
  Given an input bipartite graph $g$, we construct $m+1$ input instances $(D_1,f_1),\dots,(D_{m+1},f_{m+1})$ to our problem, where $m$ is the number of edges in $g$, in the following way. For every $r\in\set{1,\dots,m+1}$, we add one vertex $v_1$ to the left-hand side of $g$ and $r+1$ vertices $u_1,\dots,u_r,v_2$ to the right-hand side of $g$. Then, we connect the vertex $v_1$ to every new vertex on the right-hand side of $g$. We construct the instance $D_r$ from the resulting graph by adding a fact $(u,v)$ for every edge $(u,v)$ in the graph. We will compute the Shapley value of the fact $f$ corresponding to the edge $(v_1,v_2)$. The reduction is illustrated in Figure~\ref{fig:bipartite}.

In every instance $D_r$, the fact $f$ will increase the value of the measure by one in a permutation $\sigma$ if and only if $\sigma_f$ satisfies two properties: \e{(1)} the facts of $\sigma_f$ jointly satisfy the FDs in $\Delta$, and \e{(2)} $\sigma_f$ contains at least one fact that is in conflict with $f$. Hence, for $f$ to affect the value of the measure in a permutation, we have to select a set of facts corresponding to a matching from the original graph $g$, as well as exactly one of the facts corresponding to an edge $(v_1,u_i)$ (since the facts $(v_1,u_i)$ and $(v_1,u_j)$ for $i\neq j$ jointly violate the FD $A\rightarrow B$). We have the following.

$$\shapley(D_r,f,\depset,\Id)=\sum_{k=0}^m |\match(g,k)|\cdot r\cdot (k+1)!\cdot(m-k+r-1)!$$
where $\match(g,k)$ is the set of matchings of $g$ containing precisely $k$ edges.

Hence, we obtain $m+1$ equations from the $m+1$ constructed instances, and get the following system of equations.

\begin{gather*}
  \left(\!\! {\begin{array}{cccc}
        1\cdot1!m! & 1\cdot2!(m-1)! & \dots & 1\cdot(m+1)!0! \\
        2\cdot1!(m+1)! & 2\cdot2!m! & \dots & 2\cdot(m+1)!1! \\
        \vdots & \vdots & \vdots & \vdots \\
        (m+1)\cdot1!2m! & (m+1)\cdot2!(m-1)! & \dots & (m+1)\cdot(m+1)!m!
      \end{array} }\!\! \right)
  \left( {\begin{array}{c}
  |\match(g,0)| \\
   |\match(g,1)| \\
   \vdots \\
   |\match(g,m)|
 \end{array} } \right)\\
= \left(\!\! {\begin{array}{c}
      \shapley(D_1,f,\depset,\Id) \\
      \shapley(D_2,f,\depset,\Id) \\
      \vdots \\
      \shapley(D_{m+1},f,\depset,\Id)
  \end{array} }\!\! \right)
\end{gather*}

Let us divide each column in the above matrix by the constant $(j+1)!$ (where $j$ is the column number, starting from $0$) and each row by $i+1$ (where $i$ is the row number, starting from $0$), and reverse the order of the columns. We then get the following matrix.

\begin{gather*}
A=
  \left(\!\! {\begin{array}{cccc}
        0! & 1! & \dots & m! \\
        1! & 2! & \dots & (m+1)! \\
        \vdots & \vdots & \vdots & \vdots \\
        m! & (m+1)! & \dots & 2m!
      \end{array} }\!\! \right)
\end{gather*}

This matrix has coefficients $a_{i,j}=(i+j)!$, and the determinant of $A$ is $det(A)=\prod_{i=0}^{m} i!i!\neq 0$; hence, the matrix is non-singular~\cite{determinants2002}. Since dividing a column by a constant divides the determinant by a constant, and reversing the order of the columns can only change the sign of the determinant, the determinant of the original matrix is not zero as well, and the matrix is non-singular. Therefore, we can solve the system of equations, and compute the value $\sum_{k=0}^m \match(g,k)$, which is precisely the number of matchings in $g$.
\end{proof}

\paragraph{Generalization via Fact-Wise Reductions}
Using the concept of a fact-wise reduction~\cite{DBLP:conf/pods/Kimelfeld12}, we can prove hardness for any FD set that is not equivalent to an FD set with an lhs chain. We first give the formal definition of a fact-wise reduction. Let
$(R,\depset)$ and $(R',\depset')$ be two pairs of a relational schema and an FD set. A \e{mapping} from $R$ to $R'$ is a
function $\mu$ that maps facts over $R$ to facts over
$R'$. (We say that $f$ is a fact \e{over} $R$ if $f$
is a fact of some database $D$ over $R$.) We extend a mapping
$\mu$ to map databases $D$ over $R$ to databases over
$R'$ by defining $\mu(D)$ to be $\set{\mu(f)\mid f\in D}$.  A
\e{fact-wise reduction} from $(R,\depset)$ to
$(R',\depset')$ is a mapping $\Pi$ from $R$ to
$R'$ with the following properties.
\begin{enumerate}
\item $\Pi$ is injective; that is, for all facts $f$ and $g$ over
  $R$, if $\Pi(f) = \Pi(g)$ then $f = g$.
\item $\Pi$ preserves consistency and inconsistency; that is, for all facts $f$ and $g$ over
  $R$, $\set{f,g}$ satisfies $\depset$ if and only if $\set{\Pi(f),\Pi(g)}$ satisfies $\depset'$.
\item $\Pi$ is computable in polynomial time.
\end{enumerate}

We have previously shown a fact-wise reduction from $((A,B),\set{A\rightarrow B,B\rightarrow A})$ to any $(R,\depset)$, where $\depset$ is not equivalent to an FD set with an lhs chain~\cite{DBLP:conf/pods/LivshitsK17}. Clearly, fact-wise reductions preserve the Shapley value of facts, that is, $\shapley(D,f,\I,\depset)=\shapley(\Pi(D),\Pi(f),\I,\depset')$. It thus follows that there is a polynomial-time reduction from the problem of computing the Shapley value over $\set{A\rightarrow B,B\rightarrow A}$ to the problem of computing the Shapley value over any $\depset$ that has no lhs chain (even up to equivalence), and that concludes our proof of hardness.

\subsubsection{Tractability Side}
 For the tractability side of Theorem~\ref{thm:drastic}, we present a polynomial-time algorithm to compute the Shapley value. 
As stated in Observation~\ref{obs:reductionexp}, the computation of $\shapley(D,f,\depset,\Id)$ reduces in polynomial time to the computation of the expected value of the measure over all subsets of the database of a given size $m$. In this case it holds that $\Exp_{D'\sim U_m(D\setminus\set{f})}\big(\Id(D'\cup\set{f},\depset)\big)$ and $\Exp_{D'\sim U_m(D\setminus\set{f})}\big(\Id(D',\depset)\big)$ are the probabilities
that a uniformly chosen $D'\subseteq D\setminus\set{f}$ of size $m$ is such that $(D'\cup\set{f})\not\models\depset$ and $D'\not\models\depset$, respectively. Due to the structure of FD sets with an lhs chain, we can compute these probabilities efficiently, as we explain next.

Our main observation is that for an FD $X\rightarrow Y$, if we group the facts of $D$ by $X$ (i.e., split $D$ into maximal subsets of facts that agree on the values of all attributes in $X$), then this FD and the FDs that appear later in the chain may be violated only among facts from the same group. Moreover, when we group by $XY$ (i.e., further split each group of $X$ into maximal subsets of facts that agree on the values of all attributes in $Y$), facts from different groups always violate this FD, and hence, violate $\depset$.  We refer to the former groups as \e{blocks} and the latter groups as \e{subblocks}.  This special structure allows us to split the problem into smaller problems, solve each one of them separately, and then combine the solutions via dynamic programming.

We define a data structure $T$ where each vertex $v$ is associated with a subset of $D$ that we denote by $D[v]$. The root $\root$ is associated with $D$ itself, that is, $D[r]=D$. 
At the first level, each child $c$ of $r$ is associated with a block of $D[r]$ w.r.t.~$X_1\ra Y_1$, and each child $c'$ of $c$ is associated with a subblock of $D[c]$ w.r.t.~$X_1\ra Y_1$. At the second level, each child $c''$ of $c'$ is associated with a block of $D[c']$ w.r.t.~$X_2\ra Y_2$, and each child $c'''$ of $c''$ is associated with a subblock of $D[c'']$ w.r.t.~$X_2\ra Y_2$.  This continues all the way to the $n$th FD, where at the $i$th level, each child $u$ of an $(i-1)$th level subblock
vertex $v$ is associated with a block of $D[v]$ w.r.t.~$X_i\ra Y_i$ and each child
$u'$ of $u$  is associated with a subblock of $D[u]$ w.r.t.~$X_i\ra Y_i$.

We assume that the data structure $T$ is constructed in a preprocessing phase. Clearly, the number of vertices in $T$ is polynomial in $|D|$ and $n$ (recall that $n$ is the number of FDs in $\depset$) as the height of the tree is $2n$, and each level contains at most $|D|$ vertices;
hence, this preprocessing phase requires polynomial time (even under combined complexity). Then, we compute both $\Exp_{D'\sim U_m(D\setminus\set{f})}\big(\Id(D',\depset)\big)$ and $\Exp_{D'\sim U_m(D\setminus\set{f})}\big(\Id(D'\cup\set{f},\depset)\big)$ by going over the vertices of $T$ from bottom to top, as we will explain later. Note that for the computation of these values, we construct $T$ from the database $D\setminus\set{f}$. Figure~\ref{fig:T} depicts the data structure $T$ used for the computation of $\shapley(D,f_9,\depset,\Id)$ for the database $D$ and fact $f_9$ of our running example. Next, we explain the meaning of the values stored in each vertex.

  {
\begin{figure}[t]
\begin{malgorithm}{$\algname{DrasticShapley}(D,\depset,m,T)$}{alg:drastic}
\ForAll{vertices $v$ of $T$ in a bottom-up order}
\State $\algname{UpdateProb}(v,m)$
\EndFor
\State \textbf{return} $\root.\prob[m]$
\end{malgorithm}
\begin{msubroutine}{$\algname{UpdateProb}(v,m)$}{alg:updateprob}
\ForAll{children $c$ of $v$ in $T$}
\For{$j\in\set{m,\dots,1}$}
\State $v.\prob[j]\assn
\hspace{-2em}
\underset{\substack{j_1+j_2=j\\0\le j_1\le |D[c]|\\ 0\le j_2\le|D[\prev(c)]|}}{\sum}
\hspace{-2em}
\big(c.\prob[j_1]+(1-c.\prob[j_1])\cdot v.\prob[j_2]\big)\cdot \frac{{|D[c]|\choose j_1}\cdot{|D[\prev(c)]|\choose j_2}}{{{|D[\prev(c)]|+|D[c]|}\choose j}}$
\If{$v$ is a block node}
\State $v.\prob[j]\pluseq
\hspace{-2em}
\underset{\substack{j_1+j_2=j\\0<j_1\le |D[c]|\\0<j_2\le|D[\prev(c)]|}}{\sum}
\hspace{-2em}
\big((1-c.\prob[j_1])\cdot(1-v.\prob[j_2])\big)\cdot \frac{{|D[c]|\choose j_1}\cdot{|D[\prev(c)]|\choose j_2}}{{{|D[\prev(c)]|+|D[c]|}\choose j}}$
\EndIf
\EndFor
\EndFor
\end{msubroutine}
\caption{\label{alg:DrasticShapley} An algorithm for computing $\Exp_{D'\sim U_m(D\setminus\set{f})}\big(\Id(D',\depset)\big)$ for $\depset$ with an lhs chain.}
\vspace{-1.2em}
\end{figure}
}

Each vertex $v$ in $T$ stores an array $v.\prob$ with $|D[v]|+1$ entries (that is initialized with zeros) such that $v.\prob[j]=\Exp_{D'\sim U_j(D[v])}\big(\Id(D',\depset)\big)$ for all $j\in\set{0,\dots,|D[v]|}$ at the end of the execution. For this measure, we have that:
$$v.\prob[j]\eqdef\Pr{\mbox{a random subset of size } j \mbox{ of }D[v]\mbox{ violates }\depset}$$
Our final goal is to compute $\root.\prob[m]$, where $\root$ is the root of $T$. For that purpose, in the algorithm \algname{DrasticShapley}, depicted in Figure~\ref{alg:DrasticShapley}, we go over the vertices of $T$ in a bottom-up order and compute the values of $v.\prob$ for every vertex $v$ in the \algname{UpdateProb} subroutine. 
Observe that we only need one execution of \algname{DrasticShapley} with $m=|D|-1$ to compute the required values for all $m\in\set{1,\dots,|D|-1}$, as we calculate all these values in our intermediate computations. 

To compute $v.\prob$ for a subblock vertex $v$, we iterate over its children in $T$ (which are the $(i+1)$th level blocks) according to an arbitrary order defined in the construction of $T$. For a child $c$ of $v$, we denote by $\prev(c)$ the set of children of $v$ that occur before $c$ in that order, and by $D[\prev(c)]$ the database $\bigcup_{c'\in \prev(c)}D[c']$. When considering $c$ in the for loop of line~1, we compute the expected value of the measure on a subset of $D[\prev(c)]\cup D[c]$. Hence, when we consider the last child of $v$ in the for loop of line~1, we compute the expected value of the measure on a subset of the entire database $D[v]$.

For a child $c$ of $v$, there are $N_1={{{|D[\prev(c)]|+|D[c]|}\choose j}}$ subsets of size $j$ of all the children of $v$ considered so far (including $c$ itself). Each such subset consists of $j_1$ facts of the current $c$ (there are $N_2={|D[c]|\choose j_1}$ possibilities) and $j_2$ facts of the previously considered children (there are $N_3={|D[\prev(c)]|\choose j_2}$ possibilities), for some $j_1,j_2$ such that $j_1+j_2=j$, with probability
$N_2N_3/N_1$.
Moreover, such a subset violates $\depset$ if either the facts of the current $c$ violate $\depset$ (with probability $c.\prob[j_1]$ that was computed in a previous iteration) or these facts satisfy $\depset$, but the facts of the previous children violate $\depset$ (with probability $(1-c.\prob[j_1])\cdot v.\prob[j_2]$). Observe that since we go over the values $j$ in reverse order in the for loop of line~2 (i.e., from $m$ to $1$), at each iteration of this loop, we have that $v.\prob[j_2]$ (for all considered $j_2\le j$) still holds the expected value of $\Id$ over subsets of size $j_2$ of the previous children of $v$, which is indeed the value that we need for our computation.

This computation of $v.\prob$ also applies to the block vertices. However, the addition of line~5 only applies to blocks. Since the children of a block belong to different subblocks, and two facts from the same $i$th level block but different $i$th level subbblocks always jointly violate $X_i\rightarrow Y_i$, a subset of size $j$ of a block also violates the constraints if we select a non-empty subset of the current child $c$ and a non-empty subset of the previous children, even if each of these subsets by itself is consistent w.r.t.~$\depset$. Hence, we add this probability in line~5. Note that all the three cases that we consider are disjoint, so we sum the probabilities. Observe also that the leaves of $T$ have no children and we do not update their probabilities, and, indeed the probability to select a subset from a leaf $v$ that violates the constraints is zero, as all the facts of $D[v]$ agree on the values of all the attributes that occur in $\depset$.

\begin{figure}
  \centering
  \scalebox{1.0}{\input{T.pspdftex}}
  \caption{The data structure $T$ of our running example.}\label{fig:T}
  \end{figure}

\begin{exa}
We now illustrate the computation of $\Exp_{D'\sim U_m(D\setminus\set{f_9})}\big(\Id(D',\depset)\big)$ on the database $D$ and the fact $f_9$ of our running example for $m=3$. Inside each node of the data structure $T$ of Figure~\ref{fig:T}, we show the values $[v.\prob[0],v.\prob[1],v.\prob[2],v.\prob[3]]$ used for this computation. Below them, we present the corresponding values used in the computation of $\Exp_{D'\sim U_m(D\setminus\set{f_9})}\big(\Id(D'\cup\{f\},\depset)\big)$. For the leaves $v$ and each vertex $v\in\{v_5,\dots,v_9,v_{11}\}$, we have that $v.\prob[j]=0$ for every $j\in\set{0,1,2,3}$, as $D[v]$ has a single fact. As for $v_{10}$, when we consider its first child $v_{17}$ in the for loop of line~1 of \algname{UpdateProb}, all the values in $v_{10}.\prob$ remain zero (since $v_{17}.\prob[j_1]=v_{10}.\prob[j_2]=0$ for any $j_1,j_2$, and $|D[\prev(c)]|=0$). However, when we consider its second child $v_{18}$, while the computation of line~3 again has no impact on $v_{10}.\prob$, after the computation of line~5 we have that $v_{10}.\prob[2]=1$. And, indeed, there is a single subset of size two of $D[v_{10}]$, which is $\set{f_6,f_7}$, and it violates the FD $\att{train}\,\,\att{time}\,\,\att{duration}\rightarrow \att{arrives}$. This also affects the values of $v_4.\prob$. In particular, when we consider the first child $v_{10}$ of $v_4$, we have that $v_4.\prob[j]=1$ for $j=2$ and $v_4.\prob[j]=0$ for any other $j$. Then, when we consider the second child $v_{11}$ of $v_4$, it holds that $v_4.\prob[2]=\frac{1}{3}$ (as the only subset of size two of $D[v_4]$ that violates the FDs is $\set{f_6,f_7}$, and there are three subsets in total) and $v_4.\prob[3]=1$ (as every subset of size three contains both $f_6$ and $f_7$).
Finally, we have that $\Exp_{D'\sim U_3(D\setminus\set{f_9})}\big(\Id(D',\depset)\big)=\frac{55}{56}$ and $\Exp_{D'\sim U_3(D\setminus\set{f_9})}\big(\Id(D'\cup\set{f_9},\depset)\big)=1$.
\eqed\end{exa}

  {
\begin{figure}[t!]
\begin{malgorithm}{$\algname{DrasticShapleyF}(D,\depset,m,T,f)$}{alg:drasticf}
\State $\algname{DrasticShapley}(D,\depset,m,T)$
\ForAll{vertices $v$ of $T$ in a bottom-up order}
\State $\algname{UpdateProbF}(v,m,f)$
\EndFor
\State \textbf{return} $\root.\prob[m]$
\end{malgorithm}
\begin{msubroutine}{$\algname{UpdateProbF}(v,m,f)$}{alg:updateprobf}
\If{$f$ conflicts with $v$}
\State $v.\probf[j]=1$ for all $1\le j\le |D[v]|$
\State \textbf{return}
\EndIf
\If{$f$ does not match $v$}
\State $v.\probf[j]=v.\prob[j]$ for all $1\le j\le m$
\State \textbf{return}
\EndIf
\ForAll{children $c$ of $v$ in $T$}
\For{$j\in\set{m,\dots,1}$}
\State $v.\probf[j]
\assn
\hspace{-2em}
\underset{\substack{j_1+j_2=j\\0\le j_1\le |D[c]|\\ 0\le j_2\le|D[\prev(c)]|}}{\sum}
\hspace{-2em}
\big(c.\probf[j_1]+(1-c.\probf[j_1])\cdot v.\probf[j_2]\big)\cdot \frac{{|D[c]|\choose j_1}\cdot{|D[\prev(c)]|\choose j_2}}{{{|D[\prev(c)]|+|D[c]|}\choose j}}$
\If{$v$ is a block node}
\State $v.\probf[j]\pluseq
\hspace{-2em}
\underset{\substack{j_1+j_2=j\\0<j_1\le |D[c]|\\0<j_2\le|D[\prev(c)]|}}{\sum}
\hspace{-2em}
\big((1-c.\probf[j_1])\cdot(1-v.\probf[j_2])\big)\cdot \frac{{|D[c]|\choose j_1}\cdot{|D[\prev(c)]|\choose j_2}}{{{|D[\prev(c)]|+|D[c]|}\choose j}}$
\EndIf
\EndFor
\EndFor
\end{msubroutine}
\caption{\label{alg:DrasticShapleyF} An algorithm for computing $\Exp_{D'\sim U_m(D\setminus\set{f})}\big(\Id(D'\cup\set{f},\depset)\big)$ for $\depset$ with an lhs chain.}
\end{figure}
}

To compute $\Exp_{D'\sim U_m(D\setminus\set{f})}\big(\Id(D'\cup\set{f},\depset)\big)$, we use the algorithm $\algname{DrasticShapleyF}$ of Figure~\ref{alg:DrasticShapleyF}. There, we distinguish between several types of vertices w.r.t.~$f$, and show how this expectation can be computed for each one of these types.
Before elaborating on the algorithm, we give some non-standard definitions. Recall that all the facts in $D[v]$, for an $i$th level block vertex $v$, agree on the values of all the attributes in $X_1Y_1\dots X_i$. Moreover, all the facts in $D[u]$, for an $i$th level subblock vertex $u$, agree on the values of all the attributes in $X_1Y_1\dots X_iY_i$. We say that $f$ conflicts with an $i$th level block vertex $v$ if for some $X_j\rightarrow Y_j$ such that $j\in\set{1,\dots,i-1}$ it holds that $f$ agrees with the facts of $D[v]$ on all the values of the attributes in $X_j$ but disagrees with them on the attributes of $Y_j$. Note that in this case, \e{every} fact of $D[v]$ conflicts with $f$. Similarly, we say that $f$ conflicts with an $i$th level subblock vertex $u$ if it violates an FD $X_j\rightarrow Y_j$ for some $j\in\set{1,\dots,i}$ with the facts of $D[u]$.
We also say that $f$ matches an $i$th level block or subblock vertex $v$ if it agrees with the facts of $D[v]$ on the values of all the attributes in $X_1Y_1\dots X_i$.

In $\algname{DrasticShapleyF}$, we define
$v.\probf[j]$ to be the probability that a random subset of size $j$ of $D[v]$ violates $\depset$ when $f$ is added. 
We first compute $v.\prob$ for all vertices $v$ of $T$ using $\algname{DrasticShapley}$, and then we use these values to compute $v.\probf$ for some vertices $v$.
First, we observe that for vertices $v$ that conflict with $f$ we have that $v.\probf[j]=1$ for every $1\le j\le D[v]$, as every non-empty subset of $D[v]$ violates the FDs with $f$. Note that this computation also applies to the leaves of $T$ that are in conflict with $f$. For vertices $v$ that do not conflict with $f$ but also do not match with $f$, we have that $v.\probf[j]=v.\prob[j]$ for every $1\le j\le m$, as no fact of $D[v]$ agrees with $f$ on the left-hand side of an FD in $\depset$ (for $j=0$ we clearly have that $v.\probf[0]=v.\prob[0]=0$).

For the rest of the vertices, the arguments given in Section~\ref{sec:drastic} for the computation of $v.\prob$ still hold in this case; hence, the computation of $v.\probf$ is similar. In particular, for a child $c$ of $v$, a subset $E$ of size $j$ of $D[\prev(c)]\cup D[c]$ is such that $E\cup \set{f}$ violates $\depset$ if either $(E\cap D[c])\cup\set{f}$ violates $\depset$ or $(E\cap D[c])\cup\set{f}$ satisfies $\depset$ but $(E\cap D[\prev(c)])\cup\set{f}$ violates $\depset$. If $v$ is a block vertex, then $E\cup\set{f}$ also violates $\depset$ if we choose a non-empty subset from both $(E\cap D[c])$ and $(E\cap D[\prev(c)])$. Therefore, the main difference between the computation of $v.\prob$ in \algname{DrasticShapley} and the computation of $v.\probf$ in \algname{DrasticShapleyF} is the use of the value $c.\probf$ instead of the value $c.\prob$ in lines~9 and~11.

\def\drastichard{
Computing $\shapley(D,f,\depset,\Id)$ is $\fpsharpp$-complete for the FD set $\depset=\set{A\rightarrow B, B\rightarrow A}$ over the relational schema $(A,B)$.
}

\subsection{Approximation}\label{sec:drastic-approx}

We now consider an approximate computation of the Shapley value.
Using the Chernoff-Hoeffding bound, we can easily obtain an additive FPRAS of the value $\shapley(D,f,\depset,\Id)$, by sampling $O(\log(1/\delta)/\epsilon^2)$ permutations and computing the average contribution of $f$ in a permutation. As observed by Livshits et al.~\cite{DBLP:conf/icdt/LivshitsBKS20}, a multiplicative FPRAS can be obtained using the same
  algorithm
  (possibly with a different number of samples)
  if the ``gap'' property holds: nonzero Shapley values are guaranteed to be large enough compared to the utility value (which is at most $1$ in the case of the drastic measure). 
This is indeed the case here, as we now prove the following gap property of $\shapley(D,f,\depset,\Id)$.

\begin{prop}\label{prop:drastic-gap}
There is a polynomial $p$ such that for
  all databases $D$ and facts $f$ of $D$ the value $\shapley(D,f,\depset,\Id)$ is either zero or at least $1/(p(|D|))$.
\end{prop}
\begin{proof}
If no fact of $D$ is in conflict with $f$, then
  $\shapley(D,f,\depset,\Id)=0$. Otherwise, let $g$ be a fact that violates an FD of $\depset$ jointly with $f$. Clearly, it holds that $\set{g}\models\depset$, while $\set{g,f}\not\models\depset$. The probability to choose a permutation $\sigma$, such
  that $\sigma_f$ is exactly
  $\set{g}$ is $\frac{(|D|-2)!}{|D|!}=\frac{1}{|D|\cdot(|D|-1)}$ (recall that $\sigma_f$ is the set of facts that appear before $f$ in $\sigma$). Therefore, we have
  that $\shapley(D,f,\depset,\Id)\ge \frac{1}{|D|\cdot(|D|-1)}$,
  and that concludes our proof.
\end{proof}

From Proposition~\ref{prop:drastic-gap} we conclude that we can obtain an upper bound on the multiplicative 
error $\epsilon$ for 
$\shapley(D,f,\depset,\Id)$
by requiring an additive gap of $\epsilon$ divided by a polynomial. 
Hence, we get the following.

\begin{cor}\label{cor:drastic-approx}
  $\shapley(D,f,\depset,\Id)$ has both an additive and a multiplicative FPRAS.
\end{cor}

\subsection{Generalization to Multiple Relations}

We now generalize our results to schemas with multiple relation symbols. More formally, we consider (relational) schemas $\scs$ that consists of a finite set $\set{R_1,\dots,R_n}$ of relation symbols, each associated with a sequence of attributes. For a set $\Delta$ of FDs over $\scs$ and a relation symbol $R_j$ of $\scs$, we denote by $\Delta_{R_j}$ the restriction of $\Delta$ to the FDs over $R_j$. Similarly, for a database $D$ over $\scs$, we denote by $D_{R_j}$ the restriction of $D$ to the facts over $R_j$. Finally, we denote $\depset_{R_1}\cup\dots\cup \depset_{R_j}$ by $\depset^j$ and $D_{R_1}\cup\dots\cup D_{R_j}$ by $D^j$.

It is straightforward that the lower bound provided in this section also holds for schemas with multiple relation symbols. That is, given an FD set $\depset$ over a schema $\scs$, if for at least one relation symbol $R$ of $\scs$, the FD set $\depset_R$ is not equivalent to an FD set with an lhs chain, then the problem of computing $\shapley(D,f,\depset,\Id)$ is $\fpsharpp$-complete.
We now generalize our upper bound to schemas with multiple relations; that is, we focus on the case where the FD set $\depset_{R}$ of \e{every} relation symbol $R$ of the schema has an lhs chain (up to equivalence), and show that the Shapley value can be computed in polynomial time.

The formula given in  Observation~\ref{obs:reductionexp} for computing the Shapley value is general and also applies to databases over schemas with multiple relation symbols.
As aforementioned, for the drastic measure, this computation boils down to computing two probabilities---the probability that a uniformly chosen subset of $D\setminus\{f\}$ of size $m$ violates the constraints, and the probability that such a subset $D'$ satisfies $D'\cup\set{f}\not\models\Delta$. Since we consider FDs, there are no violations among facts over different relation symbols; hence, we can compute these probabilities separately for each one of the relation symbols (i.e., for every pair $(D_{R_j},\Delta_{R_j})$ of a database and its corresponding FD set), and then we combine these results using dynamic programming, as we explain next.

Let $R_1,\dots,R_n$ be an arbitrary order of the relation symbols. For each $j\in\set{1,\dots,n}$ we denote by $T_j^m$ the probability that a uniformly chosen subset of size $m$ of $D_{R_j}\setminus\set{f}$ violates $\Delta_{R_j}$. This value can be computed in polynomial time for every relation symbol, using the algorithm of Figure~\ref{alg:DrasticShapley}, as we assume that $\Delta_{R_j}$ has an lhs chain. Next, we denote by $P_j^m$ the probability that a uniformly chosen subset of size $m$ of $D^{j}\setminus\set{f}$ violates the constraints of $\Delta_{R_1}\cup\dots\cup\Delta_{R_j}$. Hence, the value $P_n^m$ is needed for the computation of the Shapley value. We compute this value using dynamic programming.
Clearly, we have that:
$$P_1^m=T_1^m$$
and for every $j>1$ we prove the following.
\begin{lem}
For $j\in\set{2,\dots,n}$ we have that:
\begin{align*}
    P_j^m=\frac{1}{{|D^j\setminus\set{f}|\choose m}}\sum_{\substack{0\le m_1\le |D_{R_j}\setminus\set{f}|\\0\le m_2\le |D^{j-1}\setminus\set{f}|\\m_1+m_2=m}}\bigg[{|D_{R_j}\setminus\set{f}|\choose m_1}&\times {|D^{j-1}\setminus\set{f}|\choose m_2}\\
    &\times
        \left(1-\left(1-T_j^{m_1}\right)\times\left(1-P_{j-1}^{m_2}\right)\right)\bigg]
\end{align*}
\end{lem}
\begin{proof}
Each subset $D'$ of size $m$ of $D^j\setminus\set{f}$ contains a subset $E_1$ of size $m_1$ of $D_{R_j}\setminus\set{f}$ and a subset $E_2$ of size $m_2$ of $D^{j-1}\setminus\set{f}$, for some $m_1,m_2$ such that $m_1+m_2=m$. Clearly, $D'$ violates the constraints if and only if at least one of $E_1$ or $E_2$ violates the constraints. That is, $\Id(D',\depset^j)=1$ if either $\Id(E_1,\depset_{R_j})=1$ or $\Id(E_2,\depset^{j-1})=1$ (or both). Therefore,
$$\Id(D',\depset^j)=1-\left(1-\Id(E_1,\depset_{R_j})\right)\times \left(1-\Id(E_2,\depset^{j-1})\right)$$
Then, we have the following:
\small
\begin{align*}
    P_j^m=
    &\Exp_{D'\sim U_m({D^{j}}\setminus\set{f})}\big(\Id(D',\depset^j)\big)
    =\sum_{\substack{D'\subseteq D^j\setminus\set{f}\\|D'|=m}}\frac{1}{{|D^{j}\setminus\set{f}|\choose m}}\Id(D',\depset^j)\\
    =&\sum_{\substack{0\le m_1\le |D_{R_j}\setminus\set{f}|\\0\le m_2\le |D^{j-1}\setminus\set{f}|\\m_1+m_2=m}}\sum_{\substack{E_1\subseteq D_{R_j}\setminus\set{f}\\
    E_2\subseteq D^{j-1}\setminus\set{f}\\|E_1|=m_1,|E_2|=m_2}}\frac{1}{{|D^{j}\setminus\set{f}|\choose m}}\times\left(1-(1-\Id(E_1,\depset_{R_j}))\times (1-\Id(E_2,\depset^{j-1}))\right)\\
    =&
    \frac{1}{{|D^{j}\setminus\set{f}|\choose m}}\sum_{\substack{0\le m_1\le |D_{R_j}\setminus\set{f}|\\0\le m_2\le |D^{j-1}\setminus\set{f}|\\m_1+m_2=m}}\Bigg[{|D_{R_j}\setminus\set{f}|\choose m_1}\times {|D^{j-1}\setminus\set{f}|\choose m_2}\times\\
    &\sum_{\substack{E_1\subseteq D_{R_j}\setminus\set{f}\\
    E_2\subseteq D^{j-1}\setminus\set{f}\\|E_1|=m_1,|E_2|=m_2}}\bigg[\frac{1}{{|D_{R_j}\setminus\set{f}|\choose m_1}}
    \times\frac{1}{{|D^{j-1}\setminus\set{f}|\choose m_2}}\times\left(1-(1-\Id(E_1,\depset_{R_j}))\times (1-\Id(E_2,\depset^{j-1}))\right)\bigg]\Bigg]\\
    =&\frac{1}{{|D^{j}\setminus\set{f}|\choose m}}\sum_{\substack{0\le m_1\le |D_{R_j}\setminus\set{f}|\\0\le m_2\le |D^{j-1}\setminus\set{f}|\\m_1+m_2=m}}\Bigg[{|D_{R_j}\setminus\set{f}|\choose m_1}\times {|D^{j-1}\setminus\set{f}|\choose m_2}\times\\
    &
    \bigg(\sum_{\substack{E_1\subseteq D_{R_j}\setminus\set{f}\\
    E_2\subseteq D^{j-1}\setminus\set{f}\\|E_1|=m_1,|E_2|=m_2}}\bigg[\frac{1}{{|D_{R_j}\setminus\set{f}|\choose m_1}}\times\frac{1}{{|D^{j-1}\setminus\set{f}|\choose m_2}}\times 1\bigg]-\\
    &\sum_{\substack{E_1\subseteq D_{R_j}\setminus\set{f}\\
    E_2\subseteq D^{j-1}\setminus\set{f}\\|E_1|=m_1,|E_2|=m_2}}\bigg[\frac{1}{{|D_{R_j}\setminus\set{f}|\choose m_1}}\times\frac{1}{{|D^{j-1}\setminus\set{f}|\choose m_2}}\times(1-\Id(E_1,\depset_{R_j}))
    \times (1-\Id(E_2,\depset^{j-1}))\bigg]\bigg)\Bigg]\\
    =&\frac{1}{{|D^{j}\setminus\set{f}|\choose m}}\sum_{\substack{0\le m_1\le |D_{R_j}\setminus\set{f}|\\0\le m_2\le |D^{j-1}\setminus\set{f}|\\m_1+m_2=m}}\Bigg[{|D_{R_j}\setminus\set{f}|\choose m_1}\times {|D^{j-1}\setminus\set{f}|\choose m_2}\times\\
    &\bigg(1-
    \left(\sum_{\substack{E_1\subseteq D_{R_j}\setminus\set{f}\\|E_1|=m_1}}\frac{1}{{|D_{R_j}\setminus\set{f}|\choose m_1}}\times(1-\Id(E_1,\depset_{R_j}))\right)\times\\
    &\left(\sum_{\substack{E_2 \subseteq D^{j-1}\setminus\set{f}\\|E_2|=m_2}}\frac{1}{{|D^{j-1}\setminus\set{f}|\choose m_2}}\times(1-\Id(E_2,\depset^{j-1}))\right)\bigg)\Bigg]\\
        =&\frac{1}{{|D^j\setminus\set{f}|\choose m}}\sum_{\substack{0\le m_1\le |D_{R_j}\setminus\set{f}|\\0\le m_2\le |D^{j-1}\setminus\set{f}|\\m_1+m_2=m}}\bigg[{|D_{R_j}\setminus\set{f}|\choose m_1}\times {|D^{j-1}\setminus\set{f}|\choose m_2}\times
        \left(1-\left(1-T_j^{m_1}\right)\times\left(1-P_{j-1}^{m_2}\right)\right)\bigg]
\end{align*}
\normalsize
This concludes our proof.
\end{proof}

\eat{
This holds since a uniformly chosen subset $E$ of size $m$ of $D^{j}$ consists of a uniformly chosen subset $E_1$ of size $m_1$ of $D_{R_j}$ and a uniformly chosen subset $E_2$ of size $m_2$ of $D^{j-1}$ for some $m_1,m_2$ such that $m_1+m_2=m$. \ester{Is this correct or should we treat each pair $(m_1,m_2)$ differently (i.e., multiply by some value that depends on $m_1$ and $m_2$)?} As the databases over different relation symbols are independent, the probability that the constraints are violated is one minus the probability that none of $E_1,E_2$ violates the constraints.}

We can similarly compute the second probability required for the Shapley value computation. The only difference is that if the fact $f$ that we consider is over the relation symbol $R_j$, then $T_j^m$ will be the probability that a uniformly chosen $D'\subset D_{R_j}$ of size $m$ is such that $D'\cup\set{f}$ violates $\Delta_{R_j}$. This value can be computed in polynomial time using the algorithm of Figure~\ref{alg:DrasticShapleyF}. Note that the results of Section~\ref{sec:drastic-approx} on the approximate computation of the Shapley value trivially generalize to schemas with multiple relation symbols; hence, there is an additive FPRAS and a multiplicative FPRAS for any set of FDs.

\section{Measure $\Imr$: The Cost of a Cardinality Repair}\label{sec:minrep}

{\color{blue}
  {
\begin{figure}[t!]
\begin{malgorithm}{$\algname{Simplify(\depset)}$}{alg:simplify1}
\State Remove trivial FDs from $\depset$
\If{$\depset$ is not empty}
\State find a removable pair $(X,Y)$ of attribute sets
\State $\Delta \assn \Delta-XY$
\EndIf
\Return $\depset$
\end{malgorithm}
\caption{\label{alg:simplify} A simplification algorithm used for deciding whether a cardinality repair w.r.t.~$\depset$ can be computed in polynomial time~\cite{DBLP:journals/tods/LivshitsKR20}.}
\end{figure}
}
}

In this section, we study the measure $\Imr$ that is based on the cost of a \e{cardinality repair}, that is, the minimal number of facts that should be deleted from the database in order to obtain a consistent subset. Unlike the other inconsistency measures considered in this article, we do not have a full dichotomy for the measure $\Imr$.

\subsection{Complexity Results}
 Livshits et al.~\cite{DBLP:journals/tods/LivshitsKR20} established a dichotomy for the problem of computing a cardinality repair, classifying FD sets into those for which the problem is solvable in polynomial time, and those for which it is NP-hard. They presented a polynomial-time algorithm, which we refer to as \algname{Simplify}, that takes as input an FD set $\depset$, finds a \e{removable} pair $(X,Y)$ of attribute sets (if such a pair exists), and removes every attribute of $X\cup Y$ from every FD in $\Delta$ (we denote the result by $\depset-XY$). A pair $(X,Y)$ of attribute sets is considered removable if it satisfies the following three conditions:
  \begin{itemize}
\item $\mathrm{Closure_{\Delta}}(X)=\mathrm{Closure_{\Delta}}(Y)$,
\item $XY$ is nonempty,
 \item  every FD in $\Delta$ contains either $X$ or $Y$ on the left-hand side.
 \end{itemize}
 Note that it may be the case that $X=Y$, and then the conditions imply that every FD of $\depset$ contains $X$ on the left-hand side. The algorithm is depicted in Figure~\ref{alg:simplify}.

Livshits et al.~\cite{DBLP:journals/tods/LivshitsKR20} have shown that if it is possible to transform $\Delta$ to an empty set by repeatedly applying $\algname{Simplify(\Delta)}$, then a cardinality repair can be computed in polynomial time. Otherwise, the problem is NP-hard (and, in fact, APX-complete).

 Fact~\ref{fact:measuretoshap} implies that computing $\shapley(D,f,\depset,\Imr)$ is hard whenever computing $\Imr(D,\depset)$ is hard. Hence, we immediately obtain the following.

\begin{thm}
Let $\depset$ be a set of FDs. If $\depset$ cannot be emptied by repeatedly applying $\algname{Simplify(\depset)}$, then computing $\shapley(D,f,\depset,\Imr)$ is NP-hard.
\end{thm}

In the remainder of this section, we focus on the tractable cases of the dichotomy of Livshits et al.~\cite{DBLP:journals/tods/LivshitsKR20}. In particular, we start by proving that the Shapley value can again be computed in polynomial time for an FD set that has an lhs chain. Note that 
FD sets with an lhs chain are a special case of FD sets that can be emptied via \algname{Simplify} steps. This holds since every FD set with an lhs chain has either an FD of the form $\emptyset\ra X$ or a set $X$ of attributes that occurs on the left-hand side of every FD. In the first case, $(\emptyset,X)$ is a removable pair, while in the second case, $(X,X)$ is a removable pair.

\begin{thm}\label{thm:minrep}
Let $\depset$ be a set of FDs. If $\depset$ is equivalent to an FD set with an lhs chain, then computing $\shapley(D,f,\depset,\Imr)$ can be done in polynomial time, given $D$ and $f$.
\end{thm}

  {
\begin{figure}[t]
\begin{malgorithm}{$\algname{RShapley}(D,\depset,m,T)$}{alg:mr}
\ForAll{vertices $v$ of $T$ in a bottom-up order}
\State $\algname{UpdateCount}(v,m)$
\EndFor
\State \textbf{return} $\sum_{k=0}^m\frac{k}{{{|D|-1}\choose m}}\cdot\root.\prob[m,k]$
\end{malgorithm}
\begin{msubroutine}{$\algname{UpdateCount}(v,m)$}{alg:updatecount}
\State $v.\prob[0,0]=1$
\IfThen{$v$ is a leaf}{$v.\prob[j,0]={|D[v]|\choose j}$ for all $j\in\set{1,\dots,|D[v]|}$}
\ForAll{children $c$ of $v$ in $T$}
\For{$j\in\set{m,\dots,1}$}
\For{$t\in\set{j,\dots,0}$}
\If{$v$ is a block vertex}
\State $v.\prob[j,t]=\underset{\substack{j_1+j_2=j\\0\le j_1\le |D[c]|\\ t-j_1\le j_2\le\min\{t,|D[\prev(c)]|\}}}{\sum}\underset{\substack{t-j_1\le w_2\le j_2}}{\sum}\big(c.\prob[j_1,t-j_2]\cdot v.\prob[j_2,w_2]\big)$
\State $v.\prob[j,t]\pluseq\underset{\substack{j_1+j_2=j\\t-j_2\le j_1\le \min\{t,|D[c]|\}\\ 0\le j_2\le|D[\prev(c)]|}}{\sum}\underset{\substack{t-j_2< w_1\le j_1}}{\sum}\big(c.\prob[j_1,w_1]\cdot v.\prob[j_2,t-j_1]\big)$
\Else
\State $v.\prob[j,t]=\underset{\substack{j_1+j_2=j\\0\le j_1\le |D[c]|\\ 0\le j_2\le|D[\prev(c)]|}}{\sum}\;\;\;\underset{\substack{t_1+t_2=t\\0\le t_1\le j_1\\ 0\le t_2\le j_2}}{\sum}\big(c.\prob[j_1,t_1]\cdot v.\prob[j_2,t_2]\big)$
\EndIf
\EndFor
\EndFor
\EndFor
\end{msubroutine}
\caption{\label{alg:MRShapley} An algorithm for computing $\Exp_{D'\sim U_m(D\setminus\set{f})}\big(\Imr(D',\depset)\big)$ for $\depset$ with an lhs chain.}
\vspace{-1.2em}
\end{figure}
}

Our polynomial-time algorithm \algname{RShapley}, depicted in Figure~\ref{alg:MRShapley}, is very similar in structure to \algname{DrasticShapley}. However, to compute the expected value of $\Imr$, we take the reduction of Observation~\ref{obs:reductionexp} a step further, and show, that the problem of computing the expected value of the measure over subsets of size $m$ can be reduced to the problem of computing the number of subsets of size $m$ of $D$ that have a cardinality repair of cost $k$, given $m$ and $k$. Recall that we refer to the number of facts that are removed from $D$ to obtain a cardinality repair $E$ as the \e{cost} of $E$. In the subroutine \algname{UpdateCount}, we compute this number. In what follows, we denote by $\MR(D,\depset)$ the cost of a cardinality repair of $D$ w.r.t.~$\depset$.

\begin{lem}\label{lemma:imr_help}
The following holds.
$$\shapley(D,f,\depset,\Imr)=\frac{1}{|D|}\sum_{m=0}^{|D|-1} \sum_{k=0}^m\frac{k}{{{|D|-1}\choose m}}|S_{m,k}^f|-|S_{m,k}|$$
where: 
$$S_{m,k}=\set{D'\subseteq D\setminus\set{f}\mid |D'|=m,\MR(D'\cup\set{f},\depset)=k}$$
$$S_{m,k}^f\set{D'\subseteq D\setminus\set{f}\mid |D'|=m,\MR(D',\depset)=k}$$
\end{lem}
\begin{proof}
We further develop the reduction of Observation~\ref{obs:reductionexp}.
\begin{align*}
  &\shapley(D,f,\depset,\Imr)=
    \frac{1}{|D|}\sum_{m=0}^{|D|-1} \underset{\substack{D'\subseteq
      (D\setminus\set{f}) \\ |D'|=m}}{\sum}
    \frac{1}{{{|D|-1}\choose m}}\Big(\I(D'\cup\set{f},\depset)-\I(D',\depset)\Big)\\
    &= \frac{1}{|D|}\sum_{m=0}^{|D|-1} \sum_{k=0}^m\underset{\substack{D'\subseteq
      (D\setminus\set{f}) \\ |D'|=m\\ \MR(D'\cup\set{f},\depset)=k}}{\sum}
    \frac{1}{{{|D|-1}\choose m}}\Big(\I(D'\cup\set{f},\depset)\Big)\\
&- \frac{1}{|D|}\sum_{m=0}^{|D|-1}  \sum_{k=0}^m\underset{\substack{D'\subseteq
      (D\setminus\set{f}) \\ |D'|=m\\ \MR(D',\depset)=k}}{\sum}
    \frac{1}{{{|D|-1}\choose m}}\Big(\I(D',\depset)\Big)\notag\\
    &= \frac{1}{|D|}\sum_{m=0}^{|D|-1} \sum_{k=0}^m\underset{\substack{D'\subseteq
      (D\setminus\set{f}) \\ |D'|=m\\ \MR(D'\cup\set{f},\depset)=k}}{\sum}
    \frac{k}{{{|D|-1}\choose m}}
- \frac{1}{|D|}\sum_{m=0}^{|D|-1}  \sum_{k=0}^m\underset{\substack{D'\subseteq
      (D\setminus\set{f}) \\ |D'|=m\\ \MR(D',\depset)=k}}{\sum}
    \frac{k}{{{|D|-1}\choose m}}\notag\\
        &= \frac{1}{|D|}\sum_{m=0}^{|D|-1} \sum_{k=0}^m\frac{k}{{{|D|-1}\choose m}}|\set{D'\subseteq D\setminus\set{f}\mid |D'|=m,\MR(D'\cup\set{f},\depset)=k}|
    \\
&- \frac{1}{|D|}\sum_{m=0}^{|D|-1}  \sum_{k=0}^m\frac{k}{{{|D|-1}\choose m}}|\set{D'\subseteq D\setminus\set{f}\mid |D'|=m,\MR(D',\depset)=k}|
    \tag*{\qedhere}\\
\end{align*}
\end{proof}

We again use the data structure $T$ defined in the previous section. For each vertex $v$ in $T$, we define:
$$v.\prob[j,t]\;\eqdef\;\mbox{number of subsets of size } j \mbox{ of }D[v]\mbox{ with a cardinality repair of cost }t$$
For the leaves $v$ of $T$, we set $v.\prob[j,0]={|D[v]|\choose j}$ for $0\le j\le |D[v]|$, as every subset of $D[v]$ is consistent, and the cost of a cardinality repair is zero. We also set $v.\prob[0,0]=1$ for each $v$ in $T$ for the same reason. Since the size of the cardinality repair is bounded by the size of the database, in $\algname{UpdateCount}(v,m)$, we compute the value $v.\prob[j,t]$ for every $1\le j\le m$ and $0\le t\le j$. To compute this number, we again go over the children of $v$, one by one. When we consider a child $c$ in the for loop of line~1, the value $v.\prob[j,t]$ is the number of subsets of size $j$ of $D[\prev(c)]\cup D[c]$ that have a cardinality repair of cost $t$. 

The children of a block $v$ are subblocks that jointly violate an FD of $\depset$; hence, when we consider a child $c$ of $v$, a cardinality repair of a subset $E$ of $D[\prev(c)]\cup D[c]$ is either a cardinality repair of $E\cap D[c]$ (in which case we remove every fact of $E\cap D[\prev(c)]$) or a cardinality repair of $E\cap D[\prev(c)]$ (in which case we remove every fact of $E\cap D[c]$). The decision regarding which of these cases holds is based on the following four parameters: \e{(1)} the number $j_1$ of facts in $E\cap D[c]$, \e{(2)} the number $j_2$ of facts in $E\cap D[\prev(c)]$, \e{(3)} the cost $w_1$ of a cardinality repair of $E\cap D[c]$, and \e{(4)} the cost $w_2$ of a cardinality repair of $E\cap D[\prev(c)]$. In particular:
\begin{itemize}
    \item If $w_1+j_2\le w_2+j_1$, then a cardinality repair of $E\cap D[c]$ is preferred over a cardinality repair of $E\cap D[\prev(c)]$, as it requires removing less facts from the database.
    \item If $w_1+j_2> w_2+j_1$, then a cardinality repair of $E\cap D[\prev(c)]$ is preferred over a cardinality repair of $E\cap D[c]$.
\end{itemize}
In fact, since we fix $t$ in the computation of $v.\prob[j,t]$, we do not need to go over all $w_1$ and $w_2$. In the first case, we have that $w_1=t-j_2$ (hence, the total number of removed facts is $t-j_2+j_2=t$), and in the second case we have that $w_2=t-j_1$ for the same reason. Hence, in line~7 we consider the first case where $t\le w_2+j_1$, and in line~8 we consider the second case where $w_1+j_2>t$. To avoid negative costs, we add a lower bound of $t-j_1$ on $j_2$ and $w_2$ in line~7, and, similarly, a lower bound of $t-j_2$ on $j_1$ and $w_1$ in line~8.

For a subblock vertex $v$, a cardinality repair of $D[v]$ is the union of cardinality repairs of the children of $v$, as facts corresponding to different children of $v$ do not jointly violate any FD. Therefore, for such vertices, in line~10, we compute $v.\prob$ by going over all $j_1,j_2$ such that $j_1+j_2=j$ and all $t_1,t_2$ such that $t_1+t_2=t$ and multiply the number of subsets of size $j_1$ of the current child for which the cost of a cardinality repair is $t_1$ by the number of subsets of size $j_2$ of the previously considered children for which the cost of a cardinality repair is $t_2$.

  {
\begin{figure}[t!]
\begin{malgorithm}{$\algname{RShapleyF}(D,\depset,m,T,f)$}{alg:mrf}
\State $\algname{RShapley(D,\depset,m,T)}$
\ForAll{vertices $v$ of $T$ in a bottom-up order}
\State $\algname{UpdateCount}(v,m,f)$
\EndFor
\State \textbf{return} $\sum_{k=0}^m\frac{k}{{{|D|-1}\choose m}}\cdot\root.\prob[m,k]$
\end{malgorithm}
\begin{msubroutine}{$\algname{UpdateCountF}(v,m,f)$}{alg:updatecountf}
\State $v.\probf[0,0]=1$
\If{$f$ conflict with $v$}
\State $v.\probf[j,t]=v.\prob[j,t-1]$ for all $1\le j\le |D[v]|$ and $1\le t\le j$
\State \textbf{return}
\EndIf
\If{$f$ does not match $v$ \textbf{or} $v$ is a leaf}
\State $v.\probf[j,t]=v.\prob[j,t]$ for all $1\le j\le m$ and $0\le t\le j$
\State \textbf{return}
\EndIf
\ForAll{children $c$ of $v$ in $T$}
\For{$j\in\set{m,\dots,1}$}
\For{$t\in\set{j,\dots,1}$}
\If{$v$ is a block vertex}
\State {\small $v.\probf[j,t]=\underset{\substack{j_1+j_2=j\\0\le j_1\le |D[c]|\\ t-j_1\le j_2\le\min\{t,|D[\prev(c)]|\}}}{\sum}\underset{\substack{t-j_1\le w_2\le j_2}}{\sum}\big(c.\probf[j_1,t-j_2]\cdot v.\probf[j_2,w_2]\big)$}
\State {\small $v.\probf[j,t]\pluseq\underset{\substack{j_1+j_2=j\\t-j_2\le j_1\le \min\{t,|D[c]|\}\\ 0\le j_2\le|D[\prev(c)]|}}{\sum}\underset{\substack{t-j_2< w_1\le j_1}}{\sum}\big(c.\probf[j_1,w_1]\cdot v.\probf[j_2,t-j_1]\big)$}
\Else
\State $v.\probf[j,t]=\underset{\substack{j_1+j_2=j\\0\le j_1\le |D[c]|\\ 0\le j_2\le|D[\prev(c)]|}}{\sum}\;\;\;\underset{\substack{t_1+t_2=t\\0\le t_1\le j_1\\ 0\le t_2\le j_2}}{\sum}\big(c.\probf[j_1,t_1]\cdot v.\probf[j_2,t_2]\big)$
\EndIf
\EndFor
\EndFor
\EndFor
\end{msubroutine}
\caption{\label{alg:MRShapleyF} An algorithm for computing $\Exp_{D'\sim U_m(D\setminus\set{f})}\big(\Imr(D'\cup\set{f},\depset)\big)$ for $\depset$ with an lhs chain.}
\vspace{-1.2em}
\end{figure}
}

Next, we give the algorithm \algname{RShapleyF} for computing $\Exp_{D'\sim U_m(D\setminus\set{f})}\big(\Imr(D'\cup\set{f},\depset)\big)$, that again involves a special treatment for vertices that conflict with $f$. The algorithm is depicted in Figure~\ref{alg:MRShapleyF}. We define:
\begin{align*}
v.\probf[j,t]\eqdef&\mbox{number of subsets of size } j \mbox{ of }D[v]\mbox{ that, jointly with }f,\\
&\mbox{ have a cardinality repair of cost }t
\end{align*}
As in the case of \algname{DrasticShapleyF}, we start with the execution of \algname{RShapley}, which allows us to reuse some of the values computed in this execution. For every vertex, we set $v.\probf[0,0]=1$, as the empty set has a single cardinality repair of cost zero. Then, we consider three types of vertices. For vertices $v$ that conflict with $f$ we have that $v.\probf[j,t]=v.\prob[j,t-1]$ for all $1\le j\le D[v]$ and $1\le t\le j$, as every non-empty subset of $v$ conflicts with $f$; hence, we have to remove $f$ in a cardinality repair, and the cost of a cardinality repair increases by one. For vertices $v$ that do not match $f$, we have that $v.\probf[j,t]=v.\prob[j,t]$, as $f$ is not in conflict with any fact of $D[v]$; hence, it can be added to any cardinality repair without increasing its cost. The same holds for the leaves of $T$ that do not conflict with $f$.

For a block vertex $v$, all the arguments given for $\algname{RShapley}$ still apply here. In particular, for a child $c$ of $v$, a cardinality repair of $E\cup\set{f}$ for a subset $E$ of size $j$ of $D[\prev(c)]\cup D[c]$, is either a cardinality repair of $(E\cap D[c])\cup\set{f}$ (in which case we delete all facts of $E\cap D[\prev(c)]$) or a cardinality repair of $(E\cap D[\prev(c)])\cup\set{f}$ (in which case we delete all facts of $E\cap D[c]$). Therefore, the only difference in the computation of $v.\probf$ compared to the computation of $v.\prob$ for such vertices is the use of $c.\probf$ (that takes $f$ into account) rather than $c.\prob$.

For a subblock vertex $v$ (that does not conflict with $f$, and, hence, matches $f$), the computation of $v.\probf$ is again very similar to that of $v.\prob$, with the only difference being the use of $c.\probf$. Observe that in this case, the children of $v$ correspond to different blocks. Each such block that does not match $f$ also does not violate any FD with $f$; hence, when we add $f$ to this block, a cardinality repair of the resulting group of facts does not require the removal of $f$. The only child of $v$ where a cardinality repair might require the removal of $f$ is a child that matches $f$, and, clearly, there is at most one such child. Therefore, we do not count the fact $f$ twice in the computation of the value $v.\probf$.

\subsection{Approximation}
In cases where a cardinality repair can be computed in polynomial time, we can obtain an additive FPRAS in the same way as  the drastic measure. (Note that this Shapley value is also in $[0,1]$.)  Moreover, we can again obtain a multiplicative FPRAS using the same technique due to the following gap property (proved similarly to Proposition~\ref{prop:drastic-gap}).  \begin{prop}
  There is a polynomial $p$ such that for
  all databases $D$ and facts $f$ of $D$ the value $\shapley(D,f,\depset,\Imr)$ is either zero or at least $1/(p(|D|))$.
\end{prop}

As aforementioned, Livshits et al.~\cite{DBLP:journals/tods/LivshitsKR20} showed that the hard cases of their dichotomy for the problem of computing a cardinality repair are, in fact, APX-complete; hence, there is a polynomial-time constant-ratio approximation, but for some $\epsilon>1$ there is no (randomized) $\epsilon$-approximation or else $\mbox{P}=\mbox{NP}$ ($\mbox{NP}\subseteq\mbox{BPP}$). Since the Shapley value of every fact w.r.t.~$\Imr$ is positive, the existence of a multiplicative FPRAS for $\shapley(D,f,\depset,\Imr)$ would imply the existence of a multiplicative FPRAS for $\Imr(D,\depset)$ (due to Fact~\ref{fact:measuretoshap}), which is a contradiction to the APX-hardness. We conclude the following.

\begin{prop}\label{prop:mr-fpras}
  Let $\depset$ be a set of FDs. If $\depset$ can be emptied by repeatedly applying $\algname{Simplify(\depset)}$, then $\shapley(D,f,\depset,\Imr)$ has both an additive and a multiplicative FPRAS. Otherwise, it has neither multiplicative nor additive FPRAS, unless $\mbox{NP}\subseteq\mbox{BPP}$. 
\end{prop}

\paragraph*{Unsolved cases for $\Imr$.}
A basic open problem is the computation of $\shapley(D,f,\depset,\Imr)$ for $\depset=\set{A\rightarrow B, B\rightarrow A}$.
On the one hand, Proposition~\ref{prop:mr-fpras} shows that this case belongs to the tractable side if an approximation is allowed.
On the other hand, our algorithm for exact  $\shapley(D,f,\depset,\Imr)$ is via counting
the subsets of size $m$ that have a cardinality repair of cost $k$. This approach will not work here:

\def\propmatchingmrhard{
Let $\depset=\set{A\rightarrow B,B\rightarrow A}$ be an FD set over $(A,B)$. Counting the subsets of size $m$ of a given database that have a cardinality repair of cost $k$ is \#P-hard.
}

\begin{prop}\label{prop:minrep-mkhard}
\propmatchingmrhard
\end{prop}
\begin{proof}
The proof is by a reduction from the problem of computing the number of perfect matchings in a bipartite graph, known to be \#P-complete~\cite{VALIANT1979189}. Given a bipartite graph $g=(A\cup B,E)$ (where $|A|=|B|$), we construct a database $D$ over $(A,B)$ by adding a fact $(a,b)$ for every edge $(a,b)\in E$. 
We then define $m=|A|$ and $k=0$. It is rather straightforward that the perfect matchings of $g$ correspond exactly to the subsets $D'$ of size $|A|$ of $D$ such that $D'$ itself is a cardinality repair.
\end{proof}

Observe that the cooperative game for  $\depset=\set{A\rightarrow B,B\rightarrow A}$ can be seen as a game on bipartite graphs where the vertices on the left-hand side represent the values of attribute $A$, the vertices on the right-hand side correspond to the values that occur in attribute $B$, and the edges represent the tuples of the database (hence, the players of the game). This game is different from the well-known matching game~\cite{DBLP:conf/stacs/AzizK14} where the players are the \e{vertices} of the graph (and the value of the game is determined by the maximum weight matching of the subgraph induced by the coalition).
In contrast, in our case the players correspond to the \e{edges} of the graph.
  It is not clear what is the connection between the two games and whether or how
we can use known results on matching games to derive results for the game that we consider here.

\eat{
\begin{proof}
The proof is by a simple reduction from the problem of computing the number of maximum matchings of a bipartite graph. Given a bipartite graph $g=(A\cup B,E)$, we construct a database $D$ over $R(A,B)$ by adding a fact $R(a,b)$ for every edge $(a,b)\in E$. Note that every maximum matching of $g$ corresponds to a cardinality repair of $D$ and vice versa.
We then compute the size $s$ of a maximum matching of $g$ (this can be done in polynomial time), and define $m=k=s$. It is straightforward that the maximum matchings of $g$ correspond exactly to the subset of size $s$ of $D$ that are cardinality repairs themselves.
\end{proof}
}
\eat{
Proposition~\ref{prop:minrep-mkhard} implies that we cannot compute $\shapley(D,f,\depset,\Imr)$ for the FD set $\depset=\set{A\rightarrow B,B\rightarrow A}$ via the expected values of the measure as we have done for FD sets with an lhs chain. Clearly, this does not imply the hardness of the problem, but it strengthens our conjecture that this is a hard problem.
Note that while we do not know the exact complexity of the computation of $\shapley(D,f,\depset,\Imr)$ for the FD set $\depset=\set{A\rightarrow B, B\rightarrow A}$, since this FD set satisfies the tractability criteria of the dichotomy of Livshits et al.~\cite{DBLP:journals/tods/LivshitsKR20} for the problem of computing a cardinality repair, we can efficiently approximate this value, according to Proposition~\ref{prop:mr-fpras}.
}

\subsection{Generalization to Multiple Relations}
As in the case of $\Imi$ and $\Ip$, 
 the results of this section easily generalize to schemas with multiple relations, due to the linearity property of the Shapley value. As in the case of the drastic measure, the (positive and negative) results on the approximate computation of the Shapley value trivially generalize to schemas with multiple relation symbols.

\section{Measure $\Imc$: The Number of Repairs}\label{sec:counting}

The final measure that we consider is $\Imc$ that counts the repairs of the database.

\subsection{Dichotomy}
A dichotomy result from our previous
work~\cite{DBLP:conf/pods/LivshitsK17} states that the problem of
counting repairs can be solved in polynomial time for FD sets with an
lhs chain (up to equivalence), and is \#P-complete for any other FD
set. The hardness side, along with Fact~\ref{fact:measuretoshap},
implies that computing $\shapley(D,f,\depset,\Imc)$ is
$\fpsharpp$-hard whenever the FD set is not equivalent to an FD set
with an lhs chain. Hence, an lhs chain is a necessary condition for
tractability. We show here that it is also sufficient: if the FD set
has an lhs chain, then the problem can be solved in polynomial
time. Consequently, we obtain the following dichotomy.

\def\theoremcounting{
Let $\depset$ be a set of FDs. If $\depset$ is equivalent to an FD set with an lhs chain, then computing $\shapley(D,f,\depset,\Imc)$ can be done in polynomial time, given $D$ and $f$. Otherwise, the problem is $\fpsharpp$-complete.}

\begin{thm}\label{thm:counting}
  \theoremcounting
\end{thm}

  {
\begin{figure}[t]
\begin{malgorithm}{$\algname{MCShapley}(D,\depset,m,T)$}{alg:mc}
\ForAll{vertices $v$ of $T$ in a bottom-up order}
\State $\algname{UpdateExpected}(v,m)$
\EndFor
\State \textbf{return} $\root.\prob[m]$
\end{malgorithm}
\begin{msubroutine}{$\algname{UpdateExpected}(v,m)$}{alg:updateexp}
\State $v.\prob[0]=1$
\IfThen{$v$ is a leaf}{$v.\prob[j]=1$ for all $j\in\set{1,\dots,|D[v]|}$}
\ForAll{children $c$ of $v$ in $T$}
\For{$j\in\set{m,\dots,1}$}
\If{$v$ is a block vertex}
\State $v.\prob[j]=\underset{\substack{j_1+j_2=j\\0\le j_1\le |D[c]|\\ 0\le j_2\le|D[\prev(c)]|}}{\sum}\big(c.\prob[j_1]+v.\prob[j_2]\big)$
\Else
\State $v.\prob[j]=\underset{\substack{j_1+j_2=j\\0\le j_1\le |D[c]|\\0\le j_2\le|D[\prev(c)]|}}{\sum}\big(c.\prob[j_1]\cdot v.\prob[j_2]\big)$
\EndIf
\EndFor
\EndFor
\end{msubroutine}
\caption{\label{alg:MCShapley} An algorithm for computing $\Exp_{D'\sim U_m(D\setminus\set{f})}\big(\Imc(D',\depset)\big)$ for $\depset$ with an lhs chain.}
\vspace{-1.2em}
\end{figure}
}
The algorithm \algname{MCShapley}, depicted in Figure~\ref{alg:MCShapley}, for computing $\shapley(D,f,\depset,\Imc)$, has the same structure as \algname{DrasticShapley}, with the only difference being the computations in the subroutine \algname{UpdateExpected} (that replaces \algname{UpdateProb}). 

For a vertex $v$ in $T$ we define:
$$v.\prob[j]=\Exp\left[\mbox{number of repairs of a random subset of size }j \mbox{ of }D[v]\right]$$
As the number of repairs of a consistent database $D$ is one ($D$ itself is a repair), we set $v.\prob[0]=1$ for every vertex $v$ and $v.\prob[j]=1$ for $0\le j\le |D[v]|$ for every leaf $v$.
Now, consider a block vertex $v$ and a child $c$ of $v$. Since the children of $v$ are subblocks, each repair consists of facts of a single child. Hence, the total number of repairs is the sum of repairs of the children of $v$. 

Using standard mathematical manipulations, we obtain the following result:
\begin{lem}\label{lemma:mc1}
For a block vertex $v$ and a child $c$ of $v$, we have that:
\begin{align*}
  &\Exp_{D'\sim U_j(D[\prev(c)]\cup D[c])}\big(\Imc(D',\depset)\big)\\
  &=\underset{\substack{j_1+j_2=j\\0\le j_1\le |D[c]|\\0\le j_2\le D[\prev(c)]}}{\sum}\hspace{-1.5em}\Exp_{D'\sim U_{j_1}(D[c])}\big(\Imc(D',\depset)\big)+\Exp_{D'\sim U_{j_2}(D[\prev(c)]}\big(\Imc(D',\depset)\big)
\end{align*}
\end{lem}
\begin{proof}
As aforementioned, each repair of a subset $E$ of $D[v]$ contains facts from a single child of $v$, and the number of repairs is the sum of repairs over the children of $v$. Moreover, since our choice of facts from different subblocks is independent, we have the following (where $\MC(D,\depset)$ is the set of repairs of $D$ w.r.t.~$\depset$).
\begin{align*}
  &\Exp_{D'\sim U_j(D[\prev(c)]\cup D[c])}\big(\Imc(D',\depset)\big)=
  \underset{\substack{D'\subseteq D[\prev(c)]\cup D[c]\\|D'|=j}}{\sum}\hspace{-2em}\Pr{D'}\cdot |\MC(D',\depset)|\\
  &= \sum_{\substack{j_1+j_2=j\\0\le j_1\le |D[c]|\\0\le j_2\le D[\prev(c)]}}\;\;\underset{\substack{E_1\subseteq D[c]\\|E_1|=j_1}}{\sum}\;\;\underset{\substack{E_2\subseteq D[\prev(c)]\\|E_2|=j_2}}{\sum}\hspace{-1em}\Pr{E_1}\Pr{E_2}\big(|\MC(E_1,\depset)|+|\MC(E_2,\depset)|\big)\\
    &= \sum_{\substack{j_1+j_2=j\\0\le j_1\le |D[c]|\\0\le j_2\le D[\prev(c)]}}\;\;\underset{\substack{E_1\subseteq D[c]\\|E_1|=j_1}}{\sum}\;\;\underset{\substack{E_2\subseteq D[\prev(c)]\\|E_2|=j_2}}{\sum}\hspace{-1em}\Pr{E_1}\Pr{E_2}|\MC(E_1,\depset)|\\
    &+\sum_{\substack{j_1+j_2=j\\0\le j_1\le |D[c]|\\0\le j_2\le D[\prev(c)]}}\;\;\underset{\substack{E_1\subseteq D[c]\\|E_1|=j_1}}{\sum}\;\;\underset{\substack{E_2\subseteq D[\prev(c)]\\|E_2|=j_2}}{\sum}\hspace{-1em}\Pr{E_1}\Pr{E_2}|\MC(E_2,\depset)|\\
    &= \sum_{\substack{j_1+j_2=j\\0\le j_1\le |D[c]|\\0\le j_2\le D[\prev(c)]}}\;\;\underset{\substack{E_1\subseteq D[c]\\|E_1|=j_1}}{\sum}\Pr{E_1}|\MC(E_1,\depset)|\;\;\underset{\substack{E_2\subseteq D[\prev(c)]\\|E_2|=j_2}}{\sum}\hspace{-1em}\Pr{E_2}\\
    &+\sum_{\substack{j_1+j_2=j\\0\le j_1\le |D[c]|\\0\le j_2\le D[\prev(c)]}}\;\;\underset{\substack{E_2\subseteq D[\prev(c)]\\|E_2|=j_2}}{\sum}\hspace{-1em}\Pr{E_2}|\MC(E_2,\depset)|\underset{\substack{E_1\subseteq D[c]\\|E_1|=j_1}}{\sum}\Pr{E_1}\\
    &= \sum_{\substack{j_1+j_2=j\\0\le j_1\le |D[c]|\\0\le j_2\le D[\prev(c)]}}\Exp_{D'\sim U_{j_1}( D[c])}\big(\Imc(D',\depset)\big)+\sum_{\substack{j_1+j_2=j\\0\le j_1\le |D[c]|\\0\le j_2\le D[\prev(c)]}}\Exp_{D'\sim U_{j_2}( D[\prev(c)])}\big(\Imc(D',\depset)\big)\\
    &= \sum_{\substack{j_1+j_2=j\\0\le j_1\le |D[c]|\\0\le j_2\le D[\prev(c)]}}\Exp_{D'\sim U_{j_1}( D[c])}\big(\Imc(D',\depset)\big)+\Exp_{D'\sim U_{j_2}( D[\prev(c)])}\big(\Imc(D',\depset)\big)
\end{align*}

Recall that in our reduction from the problem of computing the Shapley value to that of computing the expected value of the measure over subsets of a given size of the database, we considered the uniform distribution where $\Pr{E}=\frac{1}{{{|D|}\choose m}}$ for a subset $E$ of size $m$ of $D$. Therefore, we have that $\underset{\substack{E_2\subseteq D[\prev(c)]\\|E_2|=j_2}}{\sum}\Pr{E_2}=\underset{\substack{E_1\subseteq D[c]\\|E_1|=j_1}}{\sum}\Pr{E_1}=1$.
\end{proof}
The result of Lemma~\ref{lemma:mc1} is reflected in line~6 of the \algname{UpdateExpected} subroutine.
Next, we show the following result for subblock vertices, that we use for the calculation of line~8.
\begin{lem}\label{lemma:mc2}
For a subblock vertex $v$ and a child $c$ of $v$, we have that:
\begin{align*}
  &\Exp_{D'\sim U_j(D[\prev(c)]\cup D[c])}\big(\Imc(D',\depset)\big)=\\
  &\underset{\substack{j_1+j_2=j\\0\le j_1\le |D[c]|\\0\le j_2\le D[\prev(c)]}}{\sum}\hspace{-1.5em}\Exp_{D'\sim U_{j_1}(D[c])}\big(\Imc(D',\depset)\big)\cdot\Exp_{D'\sim U_{j_2}(D[\prev(c)]}\big(\Imc(D',\depset)\big)
\end{align*}
\end{lem}
\begin{proof}
Since the children of $v$ are blocks (that do not jointly violate any FD of $\depset$), each repair of a subset $E$ of $D[v]$ is a union of the repairs of the children of $v$, and the number of repairs is the product of the number of repairs over the children of $v$. Hence, we have the following:

\begin{align*}
  &\Exp_{D'\sim U_j(D[\prev(c)]\cup D[c])}\big(\Imc(D',\depset)\big)=
  \underset{\substack{D'\subseteq D[\prev(c)]\cup D[c]\\|D'|=j}}{\sum}\hspace{-2em}\Pr{D'}\cdot |\MC(D',\depset)|\\
  &= \sum_{\substack{j_1+j_2=j\\0\le j_1\le |D[c]|\\0\le j_2\le D[\prev(c)]}}\;\;\underset{\substack{E_1\subseteq D[c]\\|E_1|=j_1}}{\sum}\;\;\underset{\substack{E_2\subseteq D[\prev(c)]\\|E_2|=j_2}}{\sum}\hspace{-1em}\Pr{E_1}\Pr{E_2}\big(|\MC(E_1,\depset)|\cdot|\MC(E_2,\depset)|\big)\\
    &=\sum_{\substack{j_1+j_2=j\\0\le j_1\le |D[c]|\\0\le j_2\le D[\prev(c)]}}\;\;\underset{\substack{E_1\subseteq D[c]\\|E_1|=j_1}}{\sum}\Pr{E_1}|\MC(E_1,\depset)|\;\;\underset{\substack{E_2\subseteq D[\prev(c)]\\|E_2|=j_2}}{\sum}\hspace{-1em}\Pr{E_2}|\MC(E_2,\depset)|\\
    &= \sum_{\substack{j_1+j_2=j\\0\le j_1\le |D[c]|\\0\le j_2\le D[\prev(c)]}}\Exp_{D'\sim U_{j_1}( D[c])}\big(\Imc(D',\depset)\big)\cdot\Exp_{D'\sim U_{j_2}( D[\prev(c)])}\big(\Imc(D',\depset)\big)
    \tag*{\qedhere}
\end{align*}
\end{proof}

{
\begin{figure}[t!]
\begin{malgorithm}{$\algname{MCShapleyF}(D,\depset,m,T,f)$}{alg:mcf}
\State $\algname{MCShapley(D,\depset,m,T)}$
\ForAll{vertices $v$ of $T$ in a bottom-up order}
\State $\algname{UpdateExpectedF}(v,m,f)$
\EndFor
\State \textbf{return} $\root.\prob[m]$
\end{malgorithm}
\begin{msubroutine}{$\algname{UpdateExpectedF}(v,m,f)$}{alg:updateexpf}
\State $v.\probf[0]=1$
\If{$f$ conflict with $v$}
\State $v.\probf[j]=v.\prob[j]+1$ for all $1\le j\le |D[v]|$
\State \textbf{return}
\EndIf
\If{$f$ does not match $v$ \textbf{or} $v$ is a leaf}
\State $v.\probf[j]=v.\prob[j]$ for all $0\le j\le m$
\State \textbf{return}
\EndIf
\ForAll{children $c$ of $v$ in $T$}
\For{$j\in\set{m,\dots,1}$}
\If{$v$ is a block vertex}
\If{$c$ does not conflict with $f$}
\State $v.\probf[j]=\underset{\substack{j_1+j_2=j\\0\le j_1\le |D[c]|\\ 0\le j_2\le|D[\prev(c)]|}}{\sum}\big(c.\probf[j_1]+v.\probf[j_2]\big)$
\Else
\State $v.\probf[j]=\underset{\substack{j_1+j_2=j\\0\le j_1\le |D[c]|\\ 0\le j_2\le|D[\prev(c)]|}}{\sum}\big(c.\prob[j_1]+v.\probf[j_2]\big)$
\EndIf
\If{all the children of $v$ conflict with $f$}
\State $v.\probf[j]=v.\probf[j]+1$ for all $1\le j\le m$
\EndIf
\Else
\State $v.\probf[j]=\underset{\substack{j_1+j_2=j\\0\le j_1\le |D[c]|\\ 0\le j_2\le|D[\prev(c)]|}}{\sum}\big(c.\probf[j_1]\cdot v.\probf[j_2]\big)$
\EndIf
\EndFor
\EndFor
\end{msubroutine}
\caption{\label{alg:MCShapleyF} An algorithm for computing $\Exp_{D'\sim U_m(D\setminus\set{f})}\big(\Imc(D'\cup\set{f},\depset)\big)$ for $\depset$ with an lhs chain.}
\vspace{-1.2em}
\end{figure}
}

The algorithm \algname{MCShapleyF} that computes $\Exp_{D'\sim U_m(D\setminus\set{f})}\big(\Imc(D'\cup\set{f},\depset)\big)$ is shown in Figure~\ref{alg:MCShapleyF}. We define:
$$v.\probf[j]=\Exp\left[\mbox{number of repairs of }E\cup\set{f}\mbox{ for a random subset }E\mbox{ of size }j \mbox{ of }D[v]\right]$$
First, we set $v.\probf[0]=1$ for every vertex $v$, as when $f$ is added to the empty set we obtain a consistent database that has a single repair---the whole database. Then, we again consider three possible types of vertices. For vertices $v$ that conflict with $f$ we have that $v.\probf[j]=v.\prob[j]+1$, as $f$ violates the FDs with \e{every} non-empty subset of $D[v]$; hence, for each such subset, $\set{f}$ is an additional repair, and the number of repairs increases by one compared to the number of repairs without $f$. For a vertex $v$ that does not match $f$, it holds that $f$ does not violate the constraints with any subset of $D[v]$; thus, it does not affect the number of repairs and we have that $v.\probf[j]=v.\prob[j]$. The same holds for the leaves of $T$ that do not conflict with $f$.

For the rest of the vertices $v$, the computation is similar to the one in \algname{MCShapley}. In particular, we go over the children $c$ of $v$ in the for loop of line~8, and compute $v.\probf$ using dynamic programming. We observe that if a child $c$ of a block vertex $v$ conflicts with $f$, then when $f$ is added to a subset $E$ of $D[c]$, none of the repairs of $E\cup\set{f}$ contains $f$, but $\set{f}$ is an additional repair. For such children $c$, we use the value $c.\prob$ in the calculation of line~14, where we ignore $f$. Hence, if all the children of $v$ conflict with $f$, we compute $v.\probf$ in the exact same way we compute $v.\prob$, while ignoring the fact $f$. Then, we increase the computed value by one in line~18, to reflect the additional repair $\set{f}$.

If one of the children $c$ of $v$ does not conflict with $f$ (note that there is at most one such child), then we take $f$ into account in the computation of line~12, where we use the value $c.\probf$ rather than $c.\prob$. In this case, there is no need to increase the computed value by one, as we have already considered the addition of $f$, and the fact $f$ may appear in some of the repairs of the subset of $D[c]$ (or, again, be a repair on its own).

For a subblock vertex $v$, we compute $v.\probf$ in the same way we compute $v.\prob$ in \algname{MCShapley}, but we use the value $c.\probf$ in the computation. In this case, each repair of a subset $E$ of $D[c]\cup D[\prev(c)]$ is a union of a repair of $(E\cap D[c])\cup\set{f}$ and a repair of $(E\cap D[\prev(c)])\cup\set{f}$. Note that in this case, for a child $c$ of $v$ that does not match $f$ we have that $c.\probf=c.\prob$; hence, the fact $f$ is again only taken into account when considering a child $c$ of $v$ that matches $f$, and there is no risk in counting the repair $\set{f}$ twice.

\subsection{Approximation}

Repair counting for $\depset=\set{A\rightarrow B,B\rightarrow A}$ is
the problem of counting the maximal matchings of a bipartite graph. As
the values $\shapley(D,f,\depset,\Imc)$ are nonnegative and sum up to the
number of repairs, we conclude that an FPRAS for Shapley implies an
FPRAS for the number of maximal matchings. To the best of our
knowledge, existence of the latter is a long-standing open
problem~\cite{DBLP:journals/corr/abs-1807-04803}.  This is also the
case for any $\depset'$ that is not equivalent to an FD set with an
lhs chain, since there is a fact-wise reduction from $\depset$ to such
$\depset'$~\cite{DBLP:conf/pods/LivshitsK17}. 

\subsection{Generalization to Multiple Relations}

As in the case of the drastic measure, we can generalize the upper bound of this section to schemas with multiple relation symbols using dynamic programming. We again consider an arbitrary order $R_1,\dots,R_n$ of the realtion symbols of the schema, and denote:
$$T_j^m=\Exp_{D'\sim U_m({D_{R_j}}\setminus\set{f})}\big(\Imc(D',\depset_{R_j})\big)$$
and:
$$P_j^m=\Exp_{D'\sim U_m({D^j}\setminus\set{f})}\big(\Imc(D',\depset^j)\big)$$
The value $T_j^m$ can be computed in polynomial time, using the algorithm of Figure~\ref{alg:MCShapley}, as we assume that each $\Delta_{R_j}$ has an lhs chain. As for the value $P_j^m$, we have that $P_1^m=T_1^m$, and we prove the following for $j>1$. (Recall that we denote by $\depset^j$ the FD set $\depset_{R_1}\cup\dots\cup\depset_{R_j}$ and by $D^j$ the database $D_{R_1}\cup\dots\cup D_{R_j}$.)

\begin{lem}
For every $j\in\set{2,\dots,n}$ we have that:
\begin{align*}
 P_j^m=\frac{1}{{|{D^j}\setminus\set{f}|\choose m}}\sum_{\substack{0\le m_1\le |D_{R_j}\setminus\set{f}|\\0\le m_2\le |D^{j-1}\setminus\set{f}|\\m_1+m_2=m}}{|D_{R_j}\setminus\set{f}|\choose m_1}\times {|D^{j-1}\setminus\set{f}|\choose m_2} \times T_j^{m_1}\times P_{j-1}^{m_2}
\end{align*}
\end{lem}
\begin{proof}
A basic observation here is that the number of repairs of $D_{R_1}\cup \dots\cup D_{R_j}$ is a product of the number of repairs of $D_{R_j}$ and the number of repairs of $D_{R_1}\cup \dots\cup D_{R_{j-1}}$, since there are no conflicts among facts over different relation symbols. Thus, we have the following:
\small
\begin{align*}
    P_j^m=&\Exp_{D'\sim U_m({D^j}\setminus\set{f})}\big(\Imc(D',\depset^j)\big)
    =\sum_{\substack{D'\subseteq {D^j}\setminus\set{f}\\|D'|=m}}\frac{1}{{|{D^j}\setminus\set{f}|\choose m}}\Imc(D',\depset^j)\\
    =&\sum_{\substack{0\le m_1\le |D_{R_j}\setminus\set{f}|\\0\le m_2\le |D^{j-1}\setminus\set{f}|\\m_1+m_2=m}}\sum_{\substack{E_1\subseteq D_{R_j}\setminus\set{f}\\
    E_2\subseteq D^{j-1}\setminus\set{f}\\|E_1|=m_1,|E_2|=m_2}}\frac{1}{{|D^j\setminus\set{f}|\choose m}}\times\left(\Imc(E_1,\depset_{R_j})\times \Imc(E_2,\depset^{j-1})\right)\\
    =&\frac{1}{{|D^j\setminus\set{f}|\choose m}}\sum_{\substack{0\le m_1\le |D_{R_j}\setminus\set{f}|\\0\le m_2\le |D^{j-1}\setminus\set{f}|\\m_1+m_2=m}}\Bigg[{|D_{R_j}\setminus\set{f}|\choose m_1}\times {|D^{j-1}\setminus\set{f}|\choose m_2}\\
    &\times
    \sum_{\substack{E_1\subseteq D_{R_j}\setminus\set{f}\\
    E_2\subseteq D^{j-1}\setminus\set{f}\\|E_1|=m_1,|E_2|=m_2}}\bigg[\frac{1}{{|D_{R_j}\setminus\set{f}|\choose m_1}}
    \times\frac{1}{{|D^{j-1}\setminus\set{f}|\choose m_2}}\times\left(\Imc(E_1,\depset_{R_j})\times \Imc(E_2,\depset^{j-1})\right)\bigg]\Bigg]\\
    =&\frac{1}{{|D^j\setminus\set{f}|\choose m}}\sum_{\substack{0\le m_1\le |D_{R_j}\setminus\set{f}|\\0\le m_2\le |D^{j-1}\setminus\set{f}|\\m_1+m_2=m}}\Bigg[{|D_{R_j}\setminus\set{f}|\choose m_1}\times {|D^{j-1}\setminus\set{f}|\choose m_2}
    \\
    &\times\left(\sum_{\substack{E_1\subseteq D_{R_j}\setminus\set{f}\\|E_1|=m_1}}\frac{1}{{|D_{R_j}\setminus\set{f}|\choose m_1}}\times\Imc(E_1,\depset_{R_j})\right)\\
    &\times\left(\sum_{\substack{
    E_2\subseteq D^{j-1}\setminus\set{f}\\|E_2|=m_2}}\frac{1}{{|D^{j-1}\setminus\set{f}|\choose m_2}}\times\Imc(E_2,\depset^{j-1})\right)\Bigg]\\
    =&\frac{1}{{|D^{j}\setminus\set{f}|\choose m}}\sum_{\substack{0\le m_1\le |D_{R_j}\setminus\set{f}|\\0\le m_2\le |D^{j-1}\setminus\set{f}|\\m_1+m_2=m}}{|D_{R_j}\setminus\set{f}|\choose m_1}\times {|D^{j-1}\setminus\set{f}|\choose m_2} \times T_j^{m_1}\times P_{j-1}^{m_2}
    \tag*{\qedhere}
\end{align*}
\normalsize
\end{proof}

The computation of $\Exp_{D'\sim U_m({D^{j-1}}\setminus\set{f})}\big(\Imc(D'\cup\set{f},\depset^j)\big)$ is very similar, with the only difference being the fact that:
$$T_j^m=\Exp_{D'\sim U_m({D_{R_j}}\setminus\set{f})}\big(\Imc(D'\cup\set{f},\depset_{R_j})\big)$$
for the relation symbol $R_j$ of $f$. This value can be computed in polynomial time using the algorithm of Figure~\ref{alg:MCShapleyF}. Finally, as in the case of the drastic measure, it is rather straightforward that the lower bound of Theorem~\ref{thm:counting} generalizes to the case where the FD set $\depset_R$ has no lhs chain (up to equivalence) for at least one relation symbol $R$ of the schema.

\eat{
Fact~\ref{fact:measuretoshap}, along with the existence of a fact-wise reduction from $\depset$ to any $\depset'$ that is not equivalent to an FD set with an lhs chain,
imply that if there is an FPRAS for $\shapley(D,f,\depset',\Imc)$, then there is an FPRAS for the number of maximal matchings. To the best of our knowledge, this is a long-standing open problem~\cite{DBLP:journals/corr/abs-1807-04803}.
}

\section{Conclusions}\label{sec:conclusions}

We studied the complexity of calculating the Shapley value of database facts for basic inconsistency measures, focusing on FD constraints. We showed that two of them are computable in polynomial time: the number of violations ($\Imi$) and the number of problematic facts ($\Ip$). In contrast, each of the drastic measure ($\Id$) and the number of repairs ($\Imc$) features a dichotomy in complexity, where the tractability condition is the possession of an lhs chain (up to equivalence). For the cost of a cardinality repair ($\Imr$) we showed a tractable fragment and an intractable fragment, but a gap remains on certain FD sets---the ones that do not have an lhs chain, and yet, a cardinality repair can be computed in polynomial time. We also studied the approximability of the Shapley value and showed, among other things, an FPRAS for $\Id$ and a dichotomy in the existence of an FPRAS for $\Imr$.

\eat{
  that there is an FPRAS in the case of the drastic measure, approximation is as hard as approximating the number of maximal bipartite matchings in the case of the number of repairs (for every intractable set of FDs), and hardness of approximation for the cost of the cardinality repair.
  }

Many other directions are left open for future research.
First, the picture is incomplete for the measure $\Imr$. In particular, the complexity of the exact computation is open for the bipartite matching constraint $\set{A\ra B,B\ra A}$ that, unlike the known FD sets in the intractable fragment, has an FPRAS. In general, we would like to complete the picture of $\Imr$ towards a full dichotomy. Moreover, for the schemas where there is no FPRAS for $\Imr$, our results neither imply nor refute the existence of a constant-ratio approximation (for \e{some} constant).
Second, the problems are immediately extendible to any type of constraints other than functional dependencies, such as denial constraints, tuple generating dependencies, and so on. Third, it would be interesting to see how the results extend to wealth distribution functions other than Shapley, for instance the Banzhaff Power Index~\cite{10.2307/3689345}. The tractable cases remain tractable for the Banzhaff Power Index, but it is not clear how (and whether) our proofs for the lower bounds generalize to this function. 
Another direction is to investigate whether
properties of the database (e.g.,~bounded treewidth) have an impact on the complexity of computing the Shapley value.
Finally, there is the practical question of implementation: while our algorithms terminate in polynomial time, we believe that they are hardly scalable without further optimization and heuristics ad-hoc to the use case; developing those is an important challenge for future research.

\section*{Acknowledgment}
  \noindent This work was supported by the Israel Science
  Foundation (ISF), Grant 768/19, and the German Research Foundation (DFG) Project 412400621 (DIP program).


\bibliography{main}
\bibliographystyle{alphaurl}

\end{document}